\documentclass[preprint]{elsarticle}
\usepackage[utf8]{inputenc}
\usepackage[T1]{fontenc}
\usepackage{mathpartir}
\usepackage{enumerate}
\usepackage{paralist}
\usepackage{amssymb}
\usepackage{amsmath}
\usepackage{tikz}
\usetikzlibrary{automata,matrix,scopes,arrows,shapes}
\usepackage{xspace}
\usepackage{stmaryrd}
\usepackage[shadow]{todonotes}
\usepackage{url}

\usepackage{amsthm}
\newtheorem{theorem}{Theorem}
\newtheorem{corollary}[theorem]{Corollary}
\newtheorem{lemma}[theorem]{Lemma}
\newtheorem{proposition}[theorem]{Proposition}
\newtheorem{remark}{Remark}
\newtheorem{definition}{Definition}
\newtheorem{example}{Example}

\newcommand{\Logic}[1]{\ensuremath{\mathcal L_{#1}}}
\newcommand{\Liocos}{\ensuremath{\mathcal L_{\iocos}}}
\newcommand{\tLiocos}{\ensuremath{\widetilde{\mathcal L}_{\iocos}}}

\newcommand{\mLiocos}{\ensuremath{\mathcal L_{\iocos}^{\mu}}}
\newcommand{\nLiocos}{\ensuremath{\widetilde{\mathcal L}_{\iocos}^{\nu}}}

\newcommand{\eqLiocos}{\ensuremath{\mathcal L_{\iocos}^{\equiv}}}
\newcommand{\Lioco}{\ensuremath{\mathcal L_{\ioco}}}

\newcommand{\sem}[1]{\ensuremath{\llparenthesis #1 \rrparenthesis}}

\newcommand{\leqiocos}{\ensuremath{\leq_{\Liocos}}}
\newcommand{\tleqiocos}{\ensuremath{\leq_{\tLiocos}}}
\newcommand{\leqlogic}[1]{\ensuremath{\leq_{#1}}}
\newcommand{\false} {\ensuremath{\mathbf{ff}}\xspace}
\newcommand{\true}  {\ensuremath{\mathbf{tt}}\xspace}

\newcommand{\univ}[1] {\ensuremath{[\, #1\, ]}\xspace}
\newcommand{\uuniv}[1]{\ensuremath{\llbracket #1 \rrbracket}\xspace}
\newcommand{\exis}[1] {\ensuremath{\langle #1 \rangle}\xspace}
\newcommand{\eexis}[1]{\ensuremath{\langle\!| #1 |\!\rangle}\xspace}
\newcommand{\uioco}[1]{\ensuremath{\llfloor #1 \rrfloor}\xspace}

\newcommand{\Lcoco}{\ensuremath{\mathcal L_{CC}}}
\newcommand{\Lconf}{\ensuremath{\mathcal L_{CS}}}

\newcommand{\hml}{\ensuremath{\mbox{HML}}}
\newcommand{\mucal}{\ensuremath{\mu\mbox{-calculus}}}

\newcommand{\tra}{\mathcal{T}}
\newcommand{\Fiocos}{\mathcal{F}_{io}}

\newcommand{\coco}{\ensuremath{\lesssim_{CC}}\xspace}
\newcommand{\conf}{\ensuremath{\lesssim_{CS}}\xspace}

\newcommand{\mCRL}{\textbf{mCRL2}\xspace}

\newcommand{\LTS}{\ensuremath{\mathit{LTS}}\xspace}

\newcommand{\after}{\mathbin{\mathsf{after}}}

\newcommand{\traces}{\mathsf{traces}}

\newcommand{\outs}{\mathsf{outs}}
\newcommand{\ins}{\mathsf{ins}}

\newcommand{\Out}{\mathsf{Out}}

\makeatother

\usepackage[noend]{algpseudocode}
\usepackage{algorithm}
\algdef{SE}[CHOICE]{Choice}{EndChoice}{\algorithmicchoice}

\algnewcommand\algorithmicchoice{\textbf{choice}}
\algnewcommand\algorithmicbranch{\textbf{case}}
\algdef{SxN}[BRANCH]{Branch}{EndBranch}[1]{\algorithmicbranch\ #1\ \algorithmicdo}

\algdef{SE}[CHOICE]{Choice}{EndChoice}{\algorithmicchoice}

\newbox\arriba
\newbox\abajo
\newbox\CaracterInterno
\newbox\CaracterDerecha
\newdimen\anchura
\def\MacrosTranGeneral#1#2#3#4#5#6{%
  \setbox\CaracterInterno=\hbox{\mathsurround=0pt$\mathord#4$}
  \setbox\CaracterDerecha=\hbox{\mathsurround=0pt$\mathord#3$}
  \setbox\arriba=\hbox{$#1#2$}
  \setbox\abajo=\hbox{\mathsurround=0pt%
                      \anchura=\wd\arriba%
                      \advance \anchura by 0.5em%
                      \divide \anchura by \wd\CaracterInterno%
                      \multiply \anchura by \wd\CaracterInterno%
                      \copy\CaracterInterno\kern\SeparacionInternaFlecha
                      \hbox to \anchura{%
                          $\cleaders%
                            \hbox{\kern\SeparacionInternaFlecha\copy\CaracterInterno}
                            \hfill$}%
                      \kern\SeparacionExternaFlecha\copy\CaracterDerecha}
  \mathrel{{\buildrel\vbox{\copy\arriba \kern\SeparacionFlechaArriba} %
    \over{\copy\abajo^{#6}}}_{#5}}
  }
\def\MacrosTranGeneralProp#1#2#3#4#5{\mathchoice%
  {\MacrosTranGeneral{\scriptstyle}{#1}{#2}{#3}{#4}{#5}}
  {\MacrosTranGeneral{\scriptstyle}{#1}{#2}{#3}{#4}{#5}}
  {\MacrosTranGeneral{\scriptscriptstyle}{#1}{#2}{#3}{#4}{#5}}
  {\MacrosTranGeneral{\scriptscriptstyle}{#1}{#2}{#3}{#4}{#5}}}

\def\MacrosTran#1{%
  \def\SeparacionInternaFlecha{-0.3em}
  \def\SeparacionExternaFlecha{-0.5em}
  \def\SeparacionFlechaArriba{-3pt}
  \MacrosTranGeneralProp{#1}{\rightarrow}{-}{}{}}
\def\MacrosNoTran#1{%
  \def\SeparacionInternaFlecha{-0.3em}
  \def\SeparacionExternaFlecha{-0.5em}
  \def\SeparacionFlechaArriba{-3pt}
  \MacrosTranGeneralProp{#1\kern 0.3em}{{\not\rightarrow}}{-}{}{}}

\def\tran#1{\ensuremath\mathbin{\MacrosTran{#1}}}

\def\notran#1{\ensuremath\mathbin{\MacrosNoTran{#1}}}

\newcommand{\comen}[1]{}

\newcommand{\iocos}{\ensuremath{\mathbin{\mathsf{ioco\underline{s} }}}\xspace}
\newcommand{\niocos}{\ensuremath{\mathbin{\kern 1em\hbox to 0pt{/\hss}\kern-1em \mathsf{ioco\underline{s} }}}\xspace}
\newcommand{\iocoseq}{\ensuremath{\mathbin{\mathsf{ioco\underline{s}_{\equiv} }}}\xspace}
\newcommand{\ioco}{\ensuremath{\mathbin{\mathsf{ioco}}}\xspace}
\newcommand{\nioco}{\ensuremath{\mathbin{\kern 1em\hbox to 0pt{/\hss}\kern-1em \mathsf{ioco}}}\xspace}

\newcommand{\fail}{\lower2pt\hbox{\fontsize{8}{10}\selectfont\XSolidBrush}}

\newcommand{\nopass}{\mathbin{\kern1em\hbox to 0pt{/\hss}\kern-1em
\mathsf{pass}}}

\newcommand{\noltest}{\ensuremath{\mathbin{\kern0.5em/\kern-1em
\sqsubseteq_T}}}

\newcommand{\arr}[1]{\vec{#1}}
\newcommand{\sosrule}[2]{\frac{\raisebox{.7ex}{\normalsize{$#1$}}}{\raisebox{-1.0ex}{\normalsize{$#2$}}}}

\newcommand{\nega}{\mathsf{neg}}

\def\lparal{\mathbin{\setbox0=\hbox{$\|$}%
		\dimen0=\dp0 \advance\dimen0 -1.5pt \dp0=\dimen0%
		\underline{\kern-1.5pt\box0\kern1.5pt}}}

\def\eop{
	\ifmmode {\hbox{\Bbox}} \else \Bbox \fi
}

\journal{JLAMP}

\begin{document}
\begin{frontmatter}

\title{Logical characterisations, rule formats and compositionality
  for input-output conformance simulation\tnoteref{label1}}
\tnotetext[label1]{Research partially supported by the projects
  {Nominal SOS} (project nr.~141558-051) and Open Problems in the
  Equational Logic of Processes (project nr.~196050-051) of the
  Icelandic Research Fund; the Spanish projects DArDOS
  (TIN2015-65845-C3-1-R), TRACES (TIN2015-67522-C3-3-R) and
  Comunidad de Madrid FORTE-CM (P2018/TCS-4314);
  the project \emph{MATHADOR} (COGS
  724.464) of the European Research Council, and the Spanish addition
  to \emph{MATHADOR} (TIN2016-81699-ERC).}

\author[l1,l2]{Luca Aceto}
\ead{luca@ru.is,luca.aceto@gssi.it}
\author[l3]{Ignacio Fábregas}
\ead{ignacio.fabregas@imdea.org}
\author[l4]{Carlos Gregorio-Rodríguez}
\ead{cgr@sip.ucm.es}
\author[l2]{Anna Ing\'olfsd\'ottir}
\ead{annai@ru.is}

\address[l1]{Gran Sasso Science Institute, L'Aquila, Italy}
\address[l2]{ICE-TCS, School of Computer Science, Reykjavik University, Iceland}
\address[l3]{IMDEA Software Institute, Spain}
\address[l4]{Departamento de Sistemas Inform\'aticos y Computaci\'on, Universidad Complutense de Madrid, Spain}

\begin{abstract}
Input-output conformance simulation (\iocos) has been proposed by
Gregorio-Rodr\'iguez, Llana and Mart\'inez-Torres as a
simulation-based behavioural preorder underlying model-based
testing. This relation is inspired by Tretmans' classic {\ioco}
relation, but has better worst-case complexity than {\ioco} and
supports stepwise refinement. The goal of this paper is to develop the
theory of {\iocos} by studying logical characterisations of this
relation, rule formats for it and its compositionality. More
specifically, this article presents characterisations of {\iocos} in
terms of modal logics and compares them with an existing logical
characterisation for {\ioco} proposed by Beohar and Mousavi. It also
offers a characteristic-formula construction for {\iocos} over finite
processes in an extension of the proposed modal logics with greatest
fixed points. A precongruence rule format for {\iocos} and a rule
format ensuring that operations take quiescence properly into account
are also given. Both rule formats are based on the GSOS format by
Bloom, Istrail and Meyer. The general {modal decomposition}
methodology of Fokkink and van Glabbeek is used to show how to check
the satisfaction of properties expressed in the logic for {\iocos} in
a compositional way for operations specified by rules in the
precongruence rule format for {\iocos}.
\end{abstract}

\end{frontmatter}

\section{Introduction}\label{sec:intro}

Model-based testing (MBT) is an increasingly popular technique for
validation and verification of computing systems, and provides a
compromise between formal verification approaches, such as model
checking, and manual testing. MBT uses a model to describe the aspects
of system behaviour that are considered to be relevant at some
suitable level of abstraction. This model is employed to generate test
cases automatically, while guaranteeing that some coverage criterion
is met. Such test cases are then executed on the actual system in
order to check whether its behaviour complies with that described by
the model.

A formal notion of compliance relation between models (specifications)
and systems (implementations) provides a formal underpinning for
MBT. The de-facto standard compliance relation underlying MBT for
labelled transition systems with input and output actions is the
classic {\ioco} relation proposed by Tretmans, for which a whole MBT
framework and tools have been developed. (See, for
instance,~\cite{Tre08} and the references therein. Readers interested
in an older and very influential strand of MBT research based
on finite-state machines can find a wealth of information in the
excellent survey paper~\cite{LeeY}.)

An alternative conformance relation that can be used to underlie MBT is {\em input-output conformance simulation} 
(\iocos). This relation follows many of the ideas in the definition of {\ioco}. However, {\iocos} is a branching-time semantics 
based on simulation, whereas \ioco is a trace-based semantics. {\iocos} has been introduced, motivated and proved to be 
an adequate conformance relation for MBT in~\cite{GLM13,GLM14,GLM15}.

Since {\iocos} has been proposed as an alternative, branching-time
touchstone relation for MBT, it is natural to investigate its theory
in order to understand its properties.  The goal of this paper is to
contribute to this endeavour by studying the discriminating power of
{\iocos} and its compositionality. More precisely, in
Section~\ref{sec:basic}, we provide modal characterisations of
{\iocos} in the style of Hennessy and Milner~\cite{HM85}. We offer two
modal chacterisations of {\iocos}, which are based on the use of
either a `non-forcing diamond modality'
(Theorem~\ref{theo:logic_char}) or of a `forcing box modality'
(Theorem~\ref{theo:tlogic_char}), and compare them with existing
logical characterisations for various semantics
(Section~\ref{sec:related}) and a logic for {\ioco} proposed by Beohar
and Mousavi in~\cite{BoharMousavi14} (Section~\ref{subsec:rel_ioco}).
We also provide a characteristic formula construction for {\iocos}
(Proposition~\ref{pro:characteristic_formula} in
Section~\ref{sec:charac_formula}) for which we need an extension of
the logic for {\iocos} with fixed-points
(Section~\ref{sec:fix-points}), and show, by means of an example,
that, contrary to what is claimed in~\cite[Theorem~2]{LM13}, {\ioco}
and {\iocos} do {\em not} coincide even when implementations are input
enabled (Section~\ref{sec:relation_ioco}).

As argued in~\cite{BenesDHKN15,BijlRT03} amongst other references, MBT
can benefit from a compositional approach whose goal is to increase
the efficiency of the testing activity. The above-mentioned references
study compositionality of {\ioco} with respect to a small collection
of well-chosen operations. Here we take a general approach to the
study of compositionality of {\iocos}, which is based on the theory of
rule formats for structural operational semantics~\cite{AFV01}.  In
Section~\ref{ssec:congruence}, we present a congruence rule format for
{\iocos} based on the GSOS format proposed by Bloom, Istrail and
Meyer~\cite{BIM95} (Theorem~\ref{thm:congruence}). Some of the
conditions of the rule format for {\iocos}
(Definition~\ref{Def:ruleformat}) have similarities with those of the
rule format for $XY$-simulation given by Beohar and Mousavi
in~\cite{BeoharM15}.  However, our rule format for {\iocos} includes
  a rather involved `global' condition that stems from the fact that a
  specification need only simulate input transitions from an
  implementation that are labelled with actions that the specification
  affords. In Section~\ref{ssec:counterexamples}, we present some
  examples showing that the restrictions of the rule format from
  Definition~\ref{Def:ruleformat} cannot be relaxed easily.

  A bridge between modal characterisations of process semantics and
  rule formats for operational semantics is provided by the 
  so-called \emph{modal decomposition} method of~\cite{FokkinkGW06},
  whose roots can be traced back to the work by Larsen and
  Liu~\cite{LarsenX91}. Intuitively, this method allows one to
  determine whether some process of the form $f(p_1,\ldots, p_n)$
  satisfies a formula $\varphi$ by constructing, from $\varphi$ and
  the rules defining the operation $f$, a collection of properties
  $\varphi_1, \ldots, \varphi_n$ such that $f(p_1,\ldots,
  p_n)\models\varphi$ if, and only if, $p_i\models \varphi_i$ for
  $i\in\{1, \ldots, n\}$. This essentially amounts to providing a
  compositional model-checking procedure and proof system for the
  studied process logic with respect to operations specified by rules
  in a given format.

In Section~\ref{ssec:decomposition} we follow this approach for one of
the characterising logics for {\iocos} in order to obtain a
compositional proof system to decide whether a term built using
operations specified by rules in the {\iocos} rule format satisfies a
formula (Theorem~\ref{thm:decomposition}). The conditions of the
{\iocos} rule format play a crucial role in the proof of correctness
of the given decomposition method.

Since operations preserving {\iocos} need to take quiescence 
into account properly, we also propose a rule format guaranteeing that
operations preserve coherent quiescent behaviour
(Theorem~\ref{thm:quiescent} in Section~\ref{ssec:rule-quiescent}),
and show that it is not easy to combine the rule formats for congruence
and quiescence (Proposition~\ref{prop:merge}). Finally,
Section~\ref{sec:conclusion} concludes the paper and presents avenues
for future research.

Some of the results in Sections~\ref{sec:basic},
\ref{sec:relation_ioco} and \ref{Sect:ruleformat} were presented in
the conference paper~\cite{AFGI17}. In this extended work we offer
proofs (some in the appendices) of the results that were announced
without proof in~\cite{AFGI17} and complement the study in those
sections with new results, examples and counterexamples.  More in
detail, the present work contains the following new material.
\begin{itemize}
	\item In Section~\ref{sec:related}, we provide a comparison
          between the logics characterising {\iocos} and some
          well-known process logics.
	
	\item In Section~\ref{sec:fix-points}, we develop an extension with fixed-points of the logics for {\iocos}.
	
	\item In  Section~\ref{sec:charac_formula}, we define the characteristic formulae for \iocos-processes. 
	
	\item In Section~\ref{ssec:counterexamples}, we show by means of examples that the conditions for the congruence rule 
	format cannot be relaxed easily.
	
	\item In Section~\ref{ssec:decomposition}, we apply the modal-decomposition method for {\iocos}. 
	
	\item In Proposition~\ref{prop:merge}, we show that is not
          easy to combine the congruence and quiescence rule formats
          by considering the example of the binary merge operator.
\end{itemize}


\section{Preliminaries}
\label{sec:pre}

The input-output conformance simulation preorder presented
in~\cite{GLM13,GLM14,GLM18} (henceforth referred to as {\iocos})
is a semantic relation developed under the assumption that systems
have two kinds of actions: input actions, namely those that the
systems are willing to admit or respond to, and output actions, which
are those produced by the system and that can be seen as responses or
results.

We use $I$ to denote the alphabet of input actions, which are written
with a question mark ($a?,b?,c?\ldots$). We call $O$ the alphabet of
output actions, which are annotated with an exclamation mark
($a!,b!,\delta!$\ldots).  In many cases we want to name actions in a
general sense, inputs and outputs indistinctly.  We will consider the
set $L=I\cup O$ and we will omit the exclamation or question marks
when naming generic actions, $a,b,c \in L$.

A state with no output actions cannot proceed autonomously; such a
state is called \emph{quiescent}.
Following Tretmans
(see, for instance,~\cite{STM12,Tre08}),
we directly introduce the event of quiescence 
as a special output action denoted by $\delta!\in O$ in the definition of 
our models.

\begin{definition}\label{dfn:lts}
A \emph{labelled transition system with inputs and outputs}, LTS for short, is a quadruple $(S,I,O,\tran{})$ such that
\begin{itemize}
\item $S$ is a set of states, processes, or behaviours.
\item $I$ and $O$ are disjoint sets of input and output actions, 
  respectively. Output actions include the quiescence symbol $\delta!\in O$. 
  We define $L=I\cup O$. 
\item $\tran{} \subseteq S \times L \times S$ is the transition relation. 
  As usual, we write $p\tran{a}q$ instead of $(p,a,q)\in \tran{}$ and 
  $p\tran{a}$, for $a\in L$, if there exists some $q\in S$ such that
  $p\tran{a}q$. Analogously, we will write
  $p\notran{a}$, for $a\in L$, if there is no $q$ such that $p\tran{a} q$.

 In order to allow only for coherent quiescent systems, the set of
 transitions $\tran{}$ should also satisfy the following requirement:
 $p\tran{\delta!}p'$ iff $p=p'$ and $p\notran{a!}$ for each $a!\in
 O\setminus\{\delta!\}$.
\end{itemize}
The extension of the transition relation to sequences of actions is
defined as usual.
\end{definition}
Contrary to the classic {\ioco} testing theory, in the theory of
{\iocos} presented in \cite{GLM13,GLM14,GLM15,GLM18}, all actions are
assumed to be observable. In this paper, we follow those references
and consider only concrete actions.
 
In general we use $p,q,p',q'$\ldots\ for states or behaviours, but also $i,i',s$ and $s'$ when we want to emphasise the 
specific role of a behaviour as an implementation or a specification, respectively. We consider implementations and 
specifications, or, more generally, behaviours under study, as states of the same LTS.

The following functions over states of an LTS will be used in the remainder of the paper:
\begin{center}
  \begin{tabular}{rcl}
  $\outs(p)$&$=$&$ \{ a! \ |\ a! \in O,\ p\ \tran{a!}\}$, the set of initial outputs of a state $p$.\\
  $\ins(p)$ &$=$&$ \{ a? \ |\ a? \in I,\ p\ \tran{a?}  \} $, the set of initial
    inputs of a state $p$.
  \end{tabular}
\end{center}

The definition of input-output conformance simulation given below
stems from~\cite{GLM13,GLM14,GLM18}, to which we refer the interested
reader for motivation and discussion.
\begin{definition}\label{dfn:iocos}
  We say that a binary relation $R$ over states in an LTS is an \iocos-relation if, and only if, for each $(p,q)\in R$
  the following conditions hold:
  \begin{enumerate}
  \item \label{dfn:iocos:1} $\ins(q)\subseteq\ins(p)$. 
  \item \label{dfn:iocos:2} For all $a? \in \ins(q)$ and $p'\in S$, if $p\tran{a?} p'$ then there exists some $q'$ such that $q\tran{a?} q'$ and $(p', q')\in R$. 
  \item \label{dfn:iocos:3} For all $a! \in O$ and $p'\in S$, if $p\tran{a!} p'$ then there exists some $q'$ such that 
   $q\tran{a!} q'$ and $(p', q')\in R$. 
  \end{enumerate}
  We define the \emph{input-output conformance simulation} ($\iocos$)
  as the largest \iocos-relation, that is, the union of all \iocos-relations:
   $$\iocos =\bigcup\{ R\ |\ R\subseteq S\times S,\ R\ \text{is an \iocos-relation}\}. $$
We write $p\iocos q$ instead of $(p,q)\in \iocos$.
\end{definition}

\begin{example}\label{Ex:iocos-procs}
Consider the following processes:

{\centering
	\begin{tikzpicture}[->,>=stealth',shorten >=1pt,auto]
	\matrix [matrix of math nodes, column sep={2cm,between origins},row sep={1cm,between origins}]
	{
		\node (C0) {i}; & \node (D0) {s};\\
	};
	\begin{scope}[every node/.style={font=\small\itshape}]
	\path   
	(C0)   edge[loop right]	   node  {$a?$} (C0)
	(C0)   edge[loop left]	   node  {$\delta!$} (C0)
	(D0)   edge[loop right]     node  {$\delta!$} (D0);
	\end{scope}
	\end{tikzpicture}
	
}
It is easy to see that $i~{\iocos}~s$. Indeed, $\ins(s)=\emptyset$ and
therefore the specification $s$ does not prevent the implementation $i$
from offering the input transition $i\tran{a?} i$.
\end{example}

\begin{remark}
	In what follows, we will consider only \emph{image-finite}
        LTSs, that is, LTSs where for each $p$ and each $a\in I\cup O$
        there are only finitely many $p'$ such that
        $p\stackrel{a}{\rightarrow}p'$. Also, we will consider both
        $I$ and $O$ to be finite sets.
\end{remark}
Throughout the paper we make extensive use of modal logics. A logic
over processes is defined by a language to express the formulae in the
logic and a satisfaction relation that defines when a process (that
is, a state of an LTS) has the property described by some formula. A
classic example and a reference for the rest of the paper is
Hennessy-Milner Logic~\cite{HM85}.

\begin{definition}\label{Def:HML}
Hennessy-Milner Logic over the set of actions $L$ (abbreviated to
{\hml}) is the collection of formulae defined by the following BNF
grammar:
$$\phi ::= \true \mid \false \mid \phi \wedge \phi
\mid \phi \vee \phi \mid \univ{a}\phi\mid \exis{a}\phi, $$
where $a\in L$.  {\hml} is interpreted over an LTS by defining a
satisfaction relation $\models$ relating states to formulae. The
semantics of the boolean constants $\true$ and $\false$ and of the
boolean connectives $\wedge$ and $\vee$ is defined as usual. The
satisfaction relation for the modalities $\exis{a}$ and $\univ{a}$ is
as follows:
\begin{itemize}
\item $p\models\exis{a}\varphi$ iff there exists some $p'$ such that $p\tran{a}p'$ and $p'\models\varphi$.

\item $p\models\univ{a}\varphi$ iff for all $p'$, 
  $p\tran{a}p'$ implies $p'\models\varphi$.
\end{itemize}
As usual, we extend both boolean connectives $\wedge$ and $\vee$ to finite sets: given a finite index set $I$ and 
formulae $\varphi_i$ ($i\in I$), we define $\bigwedge_{i\in I}\varphi_i$ (respectively, $\bigvee_{i\in I} \varphi_i$) as the 
finite conjunction (respectively, finite disjunction) of the formulae $\varphi_i$. We follow the standard convention of 
considering $\bigwedge_{i\in\emptyset} \varphi_i$ equivalent to $\true$ and $\bigvee_{i\in\emptyset} \varphi_i$ 
equivalent to $\false$.
\end{definition}
Every subset of {\hml} naturally induces a preorder on a given set of
behaviours.

\begin{definition}
	\label{def:gen_logic_pre}
	Given a logic \Logic{} included in {\hml} and a set $S$ of
        states in an LTS, we define $\leqlogic{\Logic{}}$ as the
        binary relation over $S$ given by
	$$p\leqlogic{\Logic{}}q \quad \mbox{ iff }\quad \forall \phi \in \Logic{} \;\; (p\models \phi ~{\Rightarrow}~ q\models\phi).$$
\end{definition}

\begin{proposition}\label{prop:gen_logic_pre} 
For each logic \Logic{}, the binary relation $\leqlogic{\Logic{}}$ is a preorder.
\end{proposition}

\section{Logic for \iocos}
\label{sec:basic}

In this section we present logics that characterise the \iocos relations, both the preorder and its induced equivalence. 
These logics are subsets of Hennessy-Milner Logic (HML) and are convenient to characterise clearly the 
discriminating power of the \iocos relation. We will use these logics to compare \iocos\ with other behavioural 
simulation-based relations in the literature in Sections~\ref{sec:related} and~\ref{sec:relation_ioco}. In 
Section~\ref{sec:fix-points}, we will address the study of more expressive logics for \iocos that are more suitable for 
the description of system properties.

\begin{definition}
\label{def:logic}
The syntax of the logic for \iocos, denoted by \Liocos, is defined by the following BNF grammar:
$$\phi ::= \true \mid \false \mid \phi \wedge \phi
\mid \phi \vee \phi \mid \eexis{a?} \phi \mid
\exis{a!}\phi, $$
where $a?\in I$ and $a!\in O$. The semantics of the constants $\true$
and $\false$, of the boolean connectives $\wedge$ (conjunction) and $\vee$ (disjunction), and of
the modality $\exis{a!}$ (diamond) are defined as usual.  The satisfaction
relation for the modality $\eexis{a?}$ is given below:

$$p \models \eexis{a?} \phi \mbox{ iff } p\notran{a?} \mbox{ or } p' \models \phi \mbox{ for some  } p \tran{a?} p'.$$
\end{definition}

The new modal operator $\eexis{a?}$ can be read as a
\emph{non-forcing} diamond modality: if the input action labelling the
modality is not possible in a given state then the formula is
satisfied. This operator can be expressed with the classic modalities
in \hml; indeed, $\eexis{a?}\phi$ is equivalent to $\exis{a?}\phi \vee
\univ{a?}\false$. The need for this special modality arises because,
in order for $i\iocos s$ to hold, $s$ need only match the input
transitions of $i$ that are labelled with input actions that $s$
affords.
\begin{remark}
  \label{remark:equivaleceofuuniv}
	Note that the formula $\eexis{a?}\false$ is logically
        equivalent to the HML formula $\univ{a?}\false$. In other
        words, $p\models\eexis{a?}\false$ iff $p\notran{a?}$. 
\end{remark}

According to Definition~\ref{def:gen_logic_pre}, the logic \Liocos\ induces the preorder \leqiocos. Next we prove that 
this logical preorder coincides with the input-output conformance simulation preorder, \iocos, over an arbitrary 
(image-finite) LTS.

\begin{theorem}
\label{theo:logic_char}
For all states $i,s$ in some LTS, 
	$$i\iocos s \quad\mbox{iff}\quad i \leqiocos s. $$
\end{theorem}
\begin{proof}
	We prove the two implications separately. 
	\begin{itemize} 
        \item \emph{Only if implication}: Assume that $i \iocos s$ and $i\models
		\phi$ with $\phi\in\Liocos$. We show that $s\models \phi$ by
		structural induction over the formula $\phi$. We limit ourselves to
		presenting the case that $\phi=\eexis{a?}\varphi$, for some
		$\varphi$. 

        Since $s \models \eexis{a?}\varphi$ holds when
		$s\notran{a?}$, in what follows we may assume that $a?\in
		\ins(s)$. Since $i \iocos s$, by
		Definition~\ref{dfn:iocos}.\ref{dfn:iocos:1}, $\ins(s)\subseteq
		\ins(i)$. So, from $i\models \eexis{a?}\varphi$ we obtain that there
		exists some $i'$ such that $i\tran{a?}i'$ and $i'\models\varphi$.
		By Definition~\ref{dfn:iocos}.\ref{dfn:iocos:2}, we have that there
		exists some $s'$ such that $s\tran{a?} s'$ such that $i'\iocos
		s'$. Now, since $i'\models\varphi$, by the induction hypothesis, we
		have also that $s'\models \varphi$. That is,
		$s\models\eexis{a?}\varphi$, which was to be shown.
		
		\item \emph{If implication}: Consider the relation
		$R=\{(i,s)\mid i\leqiocos s\}$. We claim that $R$ is an
		$\iocos$-relation. We will prove that claim by contradiction. 

        Assume
		that $(i,s)\in R$ does not satisfy the requirements in
		Definition~\ref{dfn:iocos}. We will see that there exists a formula
		$\varphi$ in \Liocos\ such that $i\models\varphi$ and
		$s\not\models\varphi$, contradicting the assumption that $
		i\leqiocos s$.  We distinguish three cases according to the
		conditions in Definition~\ref{dfn:iocos}.
		\begin{itemize}
			\item Assume that
			$\ins(s)\not\subseteq\ins(i)$. Then, there
			exists some
			$a?\in\ins(s)\setminus\ins(i)$. In this
			case, using Remark~\ref{remark:equivaleceofuuniv},
                        $i\models\eexis{a?}\false$ but
			$s\not\models\eexis{a?}\false$.
			
			\item Assume that there exist $a?\in\ins(s)\cap \ins(i)$ and $i'$ such that $i\tran{a?} i'$ and $(i', s')\notin R$ for each $s'$ such that $s\tran{a?} s'$. Since for each  $s\tran{a?} s'$ we have $(i', s')\notin R$, there exist formulae $\varphi_{s'}$ such that $i'\models \varphi_{s'}$ and  $s'\not\models \varphi_{s'}$. Let $\varphi=\eexis{a?}\bigwedge_{s\stackrel{a?}{\rightarrow}s'}\varphi_{s'}$. By construction $i\models \varphi$ and  $s\not\models \varphi$.
			
			\item Assume that there exist $a!\in O$ and $i'$ such that $i\tran{a!} i'$ and $(i', s')\notin R$ for all $s\tran{a!} 
			s'$. As above, for each  $s\tran{a!} s'$ there exists some formula $\varphi_{s'}$ such that $i'\models 
			\varphi_{s'}$ and $s'\not\models \varphi_{s'}$. In this case, the formula 
			$\varphi=\exis{a!}\bigwedge_{s\stackrel{a!}{\rightarrow}s'}\varphi_{s'}$ is such that $i\models \varphi$ but 
			$s\not\models \varphi$. \qedhere
		\end{itemize}
	\end{itemize}
	
\end{proof}

\begin{remark}
  \label{remark:noa!}
An easy consequence of the above logical-characterisation theorem is
that there is no formula in {\Liocos} that is logically equivalent to
the HML formula $\exis{a?}\true$. Indeed, that formula is satisfied by
state $i$, but not by state $s$, in Example~\ref{Ex:iocos-procs}.
\end{remark}

The \iocos\ relation is a preorder and it induces the  equivalence relation,
\iocoseq. We can logically characterise this equivalence with the following result.

\begin{corollary}\label{cor:icoseq_Liocos}
	For all states $p,q$ in some LTS, 

	$$p\iocoseq q  \quad\mbox{iff}\quad (\forall\phi\in\Liocos\quad p\models\phi \mbox{ iff }  q\models\phi).$$
\end{corollary}

The logic for \iocos\ we have presented in Definition~\ref{def:logic}
follows a standard approach to the logical characterisation of
simulation semantics; see, for instance,~\cite{FGPR13,Gla01}.  However,
the \iocos relation originated in the model-based testing environment
where the natural reading for a logical characterisation would be
\emph{`every property satisfied by the specification should also hold 
  in the implementation.'} Next we define an
alternative logic that better matches this
specification/implementation view.

\begin{definition}\label{def:tlogic}
The syntax of the logic \tLiocos\ is defined by the following BNF grammar:

$$\phi ::= \true \mid \false \mid \phi \wedge \phi \mid \phi \vee \phi\mid \uuniv{a?} \phi \mid \univ{a!}\phi, $$ 
where $a?\in I$ and $a!\in O$. The semantics of the constants $\true$ and $\false$, of the boolean connectives 
$\wedge$ and $\vee$, and of the modality $\univ{a!}$ is defined as usual. The satisfaction relation for the
modalities $\uuniv{a?}$ is as follows:

$$p \models \uuniv{a?} \phi \mbox{ iff }p\tran{a?}\mbox{ and } p' \models \phi,\mbox{ for each  }p \tran{a?} p'.$$
\end{definition}
The new modal operator, denoted by $\uuniv{a?}$, can be read as a
\emph{forcing} box modality: the action specified in the modality must
be possible in order for a process to satisfy the formula. This
operator can be described with the classic modalities in \hml:
$\uuniv{a?}\phi$ is equivalent to $\exis{a?}\true \wedge
\univ{a?}\phi$.

\begin{remark}
Notice that the formula $\uuniv{a?}\true$ is logically equivalent to $\exis{a?}\true$. In other words, $p\models\uuniv{a?}\true$ iff $p\tran{a?}$. 	
\end{remark}

Now with this logic, we can define a preorder $\tleqiocos$ in terms of
the formulae that the specification satisfies: $s \tleqiocos i
\quad\mbox{iff}\quad \forall \phi \in \tLiocos\quad ({s\models
  \phi}\Rightarrow{i\models\phi})$.

We note that the logics $\Liocos$ and $\tLiocos$ are dual. In fact, there exist mutual transformations between both 
sets of formulae such that a behaviour satisfies one formula if, and only if, it does {\em not} satisfy the transformed 
formula. These statements are at the heart of the proof of the following result, as we will show below.

\begin{theorem}
\label{theo:tlogic_char}
	For all states $i, s$ in some LTS, $i\iocos s~\mbox{iff}~s \tleqiocos i.$

\end{theorem}

\begin{definition}
	We define the bijection $\tra:\Liocos\rightarrow\tLiocos$ by induction
	on the structure of formulae in the following way:
	\begin{itemize}
		\item $\tra(\true)=\false$.
		
		\item $\tra(\false)=\true$.
		
		\item $\tra(\phi_1 \wedge \phi_2)= \tra(\phi_1)\vee\tra(\phi_2)$.
		
		\item $\tra(\phi_1 \vee \phi_2)= \tra(\phi_1)\wedge\tra(\phi_2)$.
		
		\item $\tra(\eexis{a?} \phi)= \uuniv{a?}\tra(\phi)$. 
		
		\item $\tra(\exis{a!}\phi)= \univ{a!}\tra(\phi)$.
	\end{itemize}
	The inverse function $\tra^{-1}:\tLiocos\rightarrow\Liocos$ is defined
	in the obvious way.
\end{definition}
In order to prove Theorem~\ref{theo:tlogic_char}, we first prove the
following lemma to the effect that a behaviour satisfies some formula
$\phi$ if, and only if, it does not satisfy the transformed formula
$\tra(\phi)$.

\begin{lemma} \label{lem:logic2tlogic}
	For each state $p$ in some LTS, and formula
	$\phi\in\Liocos$, we have
	that \begin{enumerate}[(i)] \item\label{lem1} if
		$p\models\phi$ then $p\not\models\tra(\phi)$, and
		
		\item\label{lem2} if $p\not\models\phi$ then $p\models\tra(\phi)$.
	\end{enumerate}
\end{lemma}

\begin{proof}
	We prove both statements by structural induction over the formula $\phi\in\Liocos$. We limit ourselves to giving the proof for the case $\phi=\eexis{a?}\varphi$. By definition $\tra(\eexis{a?} \varphi)= \uuniv{a?}\tra(\varphi)$.
        
	(\ref{lem1}): Consider some process $p$ such that $p\models\eexis{a?} \varphi$. Since $p\models 
	\eexis{a?}\varphi$ we have that either $p\notran{a?}$, or there exists $p\tran{a?}p'$ such that	
	$p'\models\varphi$. If $p\notran{a?}$, then $p\not\models\uuniv{a?}\tra(\varphi)$. Assume now that there exists 
	some $p'$ such that	$p\tran{a?}p'$ and $p'\models\varphi$. By the inductive	hypothesis, 
	$p'\not\models\tra(\varphi)$. Therefore	$p\not\models\uuniv{a?}\tra(\varphi)$, and we are done.
	 
	(\ref{lem2}): Consider, on the other hand, some process $p$ such that $p\not\models\phi$. This means that	
	$p\tran{a?}$ and $p'\not\models\varphi$ for all	$p\tran{a?}p'$. By the induction hypothesis,	
	$p'\models\tra(\varphi)$ for each $p'$ such that $p\tran{a?}p'$. Since $p\tran{a?}$, we obtain that	
	$p\models\uuniv{a?}\tra(\varphi)$, thus concluding the proof.
\end{proof}
The following corollary will be useful in the proof of
Theorem~\ref{theo:tlogic_char}. 

\begin{corollary}\label{cor:tlogicANDlogic}
For every state $p$ in some LTS, and all formulae $\phi\in\Liocos$ and $\psi\in\tLiocos$, the following	properties hold: 
\begin{enumerate}[(i)] 
	\item $p\models\phi$ iff $p\not\models\tra(\phi)$.  
	\item $p\models\psi$ iff $p\not\models\tra^{-1}(\psi)$.  
\end{enumerate}
	
\end{corollary}
Finally, we have all the ingredients we need in order to prove
Theorem~\ref{theo:tlogic_char}, which is the equivalent result to
Theorem~\ref{theo:logic_char} for \tLiocos.

\begin{proof}[Proof of Theorem~\ref{theo:tlogic_char}]
	Assume that $i\iocos s$ and let us consider $\phi\in\tLiocos$ such that $s\models\phi$. By 
	Corollary~\ref{cor:tlogicANDlogic}.(ii), $s\not\models\tra^{-1}(\phi)$, with $\tra^{-1}(\phi)\in\Liocos$. Now, by 
	Theorem~\ref{theo:logic_char}, this implies $i\not\models\tra^{-1}(\phi)$ and by 
	Lemma~\ref{lem:logic2tlogic}.\ref{lem2}, $i\models\tra(\tra^{-1}(\phi))=\phi$.
	
	On the other hand, let us suppose that for all $\phi\in\tLiocos$, if $s\models\phi$ then $i\models\phi$. By 
	Theorem~\ref{theo:logic_char} to show that $i\iocos s$ it suffices to prove that $i\models\varphi$ implies 
	$s\models\varphi$ for each $\varphi\in\Liocos$. To this end, let us suppose that $i\models\varphi$. By 
	Corollary~\ref{cor:tlogicANDlogic}.(i), we have that $i\not\models \tra(\varphi)$. Now, by hypothesis, since 
	$\tra(\varphi)\in\tLiocos$ it must also be the case that $s\not\models \tra(\varphi)$. Applying again 
	Corollary~\ref{cor:tlogicANDlogic}.(i) we obtain that $s\models\varphi$. Thus, we conclude that $i\iocos s$, 
	finishing the proof.
\end{proof} 

As we did in Corollary~\ref{cor:icoseq_Liocos} for $\Liocos$, we can also easily characterise the equivalence relation $\iocoseq$ by using  $\tLiocos$.

\begin{corollary}\label{cor:icoseq_tLiocos}
 For all states $p,q$ in some $\LTS$,  
    $$p\iocoseq q \quad\mbox{iff}\quad (\forall\phi\in\tLiocos\quad p\models\phi \mbox{ iff } q\models\phi).$$
\end{corollary}

From Corollaries~\ref{cor:icoseq_Liocos} and~\ref{cor:icoseq_tLiocos}
we can also tell that \iocoseq can be logically characterised by means
of boolean combinations of the formulae in the logics \Liocos\ and
\tLiocos. This means that we could write formulae that are in
\Liocos\ or in \tLiocos\ and make conjunctions and disjunctions
between them, but not nest in the same formula modal operators from
\Liocos\ and \tLiocos. That is, $\eexis{a?}\true\wedge\univ{x!}\false$
would be acceptable, but $\eexis{a?}\univ{x!}\false$, would not. This is
reminiscent of the case of mutual simulation and its logical
characterisation as opposed to bisimulation equivalence and its
logical characterisation~\cite{FGPR13,Gla01}.

\begin{definition}
The logic  $\eqLiocos$ is
defined by the following BNF grammar:

$$\phi ::= \psi_{\Liocos} \mid \psi_{\tLiocos} \mid \phi\vee\phi \mid \phi\wedge\phi,$$
where $\psi_{\Liocos}\in\Liocos$ and $\psi_{\tLiocos}\in\tLiocos$.
\end{definition}

\begin{corollary}
For all states $p,q$ in some LTS, 
$$p\iocoseq q\mbox{ iff }(\forall\phi\in\eqLiocos\quad p\models\phi\mbox{ iff }q\models\phi).$$
\end{corollary}

\subsection{Related logics}
\label{sec:related}

Logics provide a systematic
way to compare behavioural relations. In this section
we relate the characterisations for \iocos\ presented in Section~\ref{sec:basic} with other logics for coinductively defined relations based on similarity, previously defined in the literature. 

First we start with HML (Definition~\ref{Def:HML}), characterising bisimulation equivalence~\cite{HM85}, 
that 
determines a clear upper bound on the expressiveness of the logics for \iocos. In what follows, (strict) inclusions and 
equalities between logics are up to logical equivalence. For example, ${\mathcal L}\subseteq {\mathcal L'}$ 
means 
that, for each $\varphi\in{\mathcal L}$, there is some $\varphi'\in{\mathcal L'}$ that is logically equivalent to 
$\varphi$. 

\begin{proposition}
$\Liocos \subset  \hml$ and  $\tLiocos \subset  \hml$.
\begin{proof}
  The only non-standard operator used in \Liocos\ is the
  $\eexis{\cdot}$ operator ($\uuniv{\cdot}$ in $\tLiocos$,
  respectively), but, as observed earlier, formulae
  $\eexis{a?}\varphi$ ($\uuniv{a?}\varphi$, respectively) can be
  expressed in \hml\ as $\exis{a?}\varphi \vee \univ{a?}\false$
  ($\exis{a?}\true \wedge \univ{a?}\varphi$, respectively).  We assume
  no distinction between input and output actions in \hml. As noted
  before in Remark~\ref{remark:noa!}, the HML formula $\exis{a?}\true$
  cannot be expressed in {\Liocos} and thus $\tLiocos$ cannot express
  its dual $\univ{a?}\false$.
\end{proof}
\end{proposition}
 
The simulation and ready simulation preorders~\cite{BIM95} have been logically characterised in the 
literature using subsets of \hml, see for instance~\cite{BM92,FGPR13,Gla01}.

\begin{definition}\label{def:sim_rs}
The logics for plain simulation and ready simulation, denoted by $\Logic{s}$ and $\Logic{rs}$ respectively, 
are defined by the following BNF grammars, where $a\in L$:
 
$$
 \begin{array}{lcl}
 \Logic{s}&:&\phi ::= \false \mid \true \mid \phi \wedge \phi \mid \phi \vee \phi \mid \exis{a}\phi. \\
 \Logic{rs}&:& \phi ::= \false \mid \true \mid \phi \wedge \phi \mid \phi \vee \phi \mid \exis{a}\phi \mid \univ{a}\false. 
\end{array}
$$
\end{definition}

\begin{remark}\label{rem:disjunction}
In the classic literature, see for instance~\cite{BIM95,BM92,Gla01},
the logics characterising simulation and ready simulation include
neither the disjunction operator $\vee$ nor the constant $\false$;
however, as was already proved in~\cite{FGPR13}, those constructs can
be safely added to those logics without altering their discriminating
power.
\end{remark}

\begin{proposition}\label{prop:sim_rs}  
For the logics $\Logic{s}, \Logic{rs}$ and $\Liocos$ the following properties hold:
  \begin{enumerate}
  \item $\Logic{rs} \supseteq \Logic{\iocos}$,
  \item $\Logic{rs} \not\subseteq \Logic{\iocos}$,
  \item $\Logic{s} \not\supseteq \Logic{\iocos}$ and
  \item $\Logic{s} \not\subseteq \Logic{\iocos}$.
  \end{enumerate}
\end{proposition}

  \begin{proof}
	Let us consider a state $p$ with transitions $p\tran{a?}p$, $p\tran{b?}p$ and $p\tran{\delta!}p$, and state 
	$q$ with the transitions $q\tran{a?}q$ and $q\tran{\delta!}q$. Observe that $p \iocos q$ and therefore 
	$\Liocos(p)\subseteq \Liocos(q)$ (Theorem~\ref{theo:logic_char}). This will be used in the proofs of statements 2--4 below.
	We consider each statement separately. 
	\begin{enumerate}
		\item $\Logic{rs} \supseteq \Logic{\iocos}$. All formulas in $\Logic{\iocos}$ are also in $\Logic{rs}$.
		
		\item $\Logic{rs} \not\subseteq \Logic{\iocos}$. Considering states $p$ and $q$ above, we 
		have that  $p\iocos q$ and therefore $p\leqlogic{\Liocos} q$. However, the formula $\exis{b?}\true$ in 
		$\Logic{rs}$ is	satisfied by $p$ and not $q$. So there is no \Liocos\ formula equivalent to it.
		
		\item $\Logic{s} \not\supseteq \Logic{\iocos}$. Considering states $p$ and $q$ above, we have that 
		$q$ is simulated by $p$, that is $q\leqlogic{\Logic{s}}p$, but formula $\eexis{b?}\false$ is satisfied by 
		$q$ and not by $p$ and therefore $q\not \leqlogic{\Logic{\iocos}}p$. It follows that there is no formula 
		in $\Logic{s}$ that is logically equivalent to $\eexis{b?}\false$.
		
		\item $\Logic{s} \not\subseteq \Logic{\iocos}$. Considering again states $p$ and $q$. We have that 
		$p\iocos q$ and therefore $p\leqlogic{\Logic{\iocos}}q$, but $p\not\leqlogic{\Logic{s}}q$ because $p$ 
		satisfies $\exis{a?}\true$ but $q$ does not. This means that $\exis{a?}\true$ cannot be expressed in 
		$\Liocos$ up to logical equivalence. \qedhere
	\end{enumerate}
\end{proof}

As we already discussed in Remark~\ref{rem:disjunction}, disjunction
does not add any distinguishing power to $\Logic{rs}$. However,
disjunction is needed for the validity of statement $1$ in
Proposition~\ref{prop:sim_rs}, as the next lemma formalises.

\begin{lemma}
Let $\Logic{rs}'$ be defined as 
$$\Logic{rs}': \phi ::= \true \mid \phi \wedge \phi \mid \exis{a}\phi \mid \univ{a}\false , $$
where $a\in L$. Then, there is no formula in $\Logic{rs}'$ such that is logically equivalent to the formula
$\eexis{a?}\exis{b!}\true \in \Liocos$.
\end{lemma}

\begin{proof}
Let $\phi$ be a formula in $\Logic{rs}'$. We will show that $\phi$ is not equivalent to 
$\psi=\eexis{a?}\exis{b!}\true$. To this end, note first of all that, up to logical equivalence, $\phi$ can be 
written as $\phi=\bigwedge_{i\in I} \exis{a_i}\phi_i\wedge\bigwedge_{j\in J}\univ{b_j}\false$, for finite sets $I$ 
and $J$, actions $a_i$ and $b_j$ and formulae $\phi_i$. Notice that if $I = J = \emptyset$, then $\phi$ is a 
tautology which is not logically equivalent to $\psi=\eexis{a?}\exis{b!}\true$. Let us assume then that $I\neq 
\emptyset$ or $J\neq \emptyset$. We distinguish two cases.

First, assume that some $b_j$ is different from $a?$. Let us consider
the process $p$ such that $p\tran{b_j}p$ and $p\tran{\delta!}p$. We
have that $p\models \psi$ but $p\not\models\phi$, which implies that
$\psi$ and $\phi$ are not logically equivalent.
	
Next, assume that all the $b_j$'s are equal to $a?$. There are two subcases.
	\begin{itemize}
	\item $J\neq \emptyset$. In this case, if there exists some
          $i\in I$ such that $a_i=a?$, $\phi$ is unsatisfiable, and
          thus is not equivalent to $\psi$. Otherwise, let us consider
          two subcases. If $I=\emptyset$, $\phi$ is logically
          equivalent to $\univ{a?}\false$ which is not equivalent to
          $\psi$. Otherwise, if $I\neq\emptyset$, the process $p$ such
          that $p\tran{a?}p$ and $p\tran{b!}p$, clearly satisfies
          $\psi$ but $p\not\models \phi$.
	
	\item $J = \emptyset$. In this case, $I\neq \emptyset$. There are two cases: either there exists some 
	$i\in I$ such that $a_i=a?$, or there is not. In the former, the process $\text{nil} \tran{\delta!} \text{nil}$ satisfies $\psi$ but does not satisfy $\phi$. In 
	the latter, let us consider a process $q$ such that $q\tran{a?}q$ and $q\tran{\delta!}q$. Clearly $q\models\psi$ but 
	$q\not\models \phi$, which implies that those formulae are not equivalent.\qedhere  	
	\end{itemize}
\end{proof}

Besides the classic semantic relations (bisimulation, simulation and
ready simulation) there is a less standard relation in the literature,
the so-called covariant-contravariant simulation, that is quite close
to \iocos\ in its formulation. Let us recall here this relation and
its characterising logic.

\begin{definition}[Covariant-contravariant simulation~\cite{FabregasFP10}]
\label{co_contra_sim}
Consider an alphabet $A$, and let $\{A^r,A^l, A^{\mathit{bi}}\}$ be a partition of this alphabet. A binary 
relation $R$ over states in an LTS is a covariant-contravariant simulation relation if and only if for each 
$(p,q)\in R$ the following conditions hold:
\begin{itemize}
\item for all $a\in A^r\cup A^{\mathit{bi}}$ and all $p\tran{a}p'$ there exists $q\tran{a}q'$ with $(p',q')\in R$,
\item for all $b\in A^l\cup A^{\mathit{bi}}$, and all $q\tran{b}q'$ there exists $p\tran{b}p'$ with $(p',q')\in R$.
\end{itemize}
We will write $p\coco q$ if there exists a covariant-contravariant simulation $R$ such that $(p,q)\in R$.
\end{definition}

Let us note that the case of $A^l =  A^{\mathit{bi}} = \emptyset$ and $A^r= L$ would yield plain simulation. In 
addition, the case  $A^l =  A^r = \emptyset$ and $A^{\mathit{bi}}= L$ would yield bisimulation.  
Covariant-contravariant simulation are also called $XY$-simulation (for instance in \cite{AV10}). There the 
sets $X$ and $Y$ might not be disjoint, but the definitions are equivalent.

\begin{definition}[Covariant-contravariant logic~\cite{FabregasFP10}]
The syntax of the logic for covariant-contravariant, denoted by \Lcoco, is defined by the following BNF grammar:.

$$\phi ::= \true \mid \false \mid \phi \wedge \phi
\mid \phi \vee \phi \mid \exis{a} \phi \mid
\univ{b}\phi,$$
where $a\in A^r\cup A^{\mathit{bi}}$ and $b\in A^l\cup A^{\mathit{bi}}$.
\end{definition}

There are clear similarities between Definition~\ref{co_contra_sim}
and the definition of {\iocos} in Definition~\ref{dfn:iocos}.
Nevertheless, as it is not difficult to prove from the characterising
logics, the relations \coco\ and \iocos\ are not related.
 
\begin{proposition}
The covariant-contravariant simulation and \iocos\ are not comparable.
\end{proposition}
\begin{proof}
Given that for the \iocos\ relation the input and output sets are disjoint, we consider  
$A^{\mathit{bi}}=\emptyset$. Let us consider the following processes:
\begin{itemize}
\item $p\tran{a?}p'$, $p\tran{\delta!}p$ and $p'\tran{b!}p'$;

\item  $q\tran{a?}q_1$, $q\tran{\delta!}q$, $q\tran{a?}q_2$, $q_1\tran{b!}q_1$ and $q_2\tran{\delta!}q_2$;

\item $r\tran{a?}r$, $r\tran{b?}r$ and $r\tran{\delta!}r$; and

\item $s\tran{a?}s$ and $s\tran{\delta!}s$.
\end{itemize}

We discuss the different alternatives:
\begin{itemize}  
\item If  $A^r=I$, $A^l = O$, then $\Lcoco$ can be expressed as
 $$\phi ::= \true \mid \false \mid \phi \wedge \phi \mid \phi \vee \phi \mid \exis{a?} \phi \mid \univ{b!}\phi,$$

which is not comparable with $\tLiocos$
$$\phi ::= \true \mid \false \mid \phi \wedge \phi \mid \phi \vee \phi \mid \uuniv{a?} \phi \mid \univ{b!}\phi. $$

First, let us consider the processes $p$ and $q$. $p\coco q$ which implies that 
$\Lcoco(p)\subseteq\Lcoco(q)$, but $p\models \uuniv{a?}\univ{\delta!}\false$ whereas $q\not\models 
\uuniv{a?}\univ{\delta!}\false$. It follows that there is no formula in $\Lcoco$ that is logically equivalent to 
$\uuniv{a?}\univ{\delta!}\false$.

Analogously, consider $r$ and $s$. Is clear that $r \iocos s$ which implies that 
$\tLiocos(r)\subseteq\tLiocos(s)$, but $r\models\exis{b?}\true$ and $s\not\models\exis{b?}\true$, which 
implies that there is no formula in $\tLiocos$ that is logically equivalent to $\exis{b?}\true$.

\item If  $A^r=O$, $A^l = I$, then $\Lcoco$ can be expressed as
$$\phi ::= \true \mid \false \mid \phi \wedge \phi \mid \phi \vee \phi \mid \univ{a?}\phi \mid \exis{b!} \phi,$$

that is not comparable to $\Liocos$
$$\phi ::= \true \mid \false \mid \phi \wedge \phi \mid \phi \vee \phi \mid \eexis{a?} \phi \mid\exis{b!}\phi. $$

First, let us consider $p$ and $q$. Observe that $p\iocos q$, which implies 
$\Liocos(p)\subseteq \Liocos(q)$. Now, $p\models \univ{a?}\exis{b!}\true$ but $q\not\models 
\univ{a?}\exis{b!}\true$. It follows that there is no formula in $\Lioco$ that is logically equivalent to 
$\univ{a?}\exis{b!}\true$.

Analogously, consider $r$ and $s$. Is clear that $s \coco r$ which implies that 
$\Lcoco(s)\subseteq\Lcoco(r)$, but $s\models\eexis{b?}\univ{\delta!}\false$ and 
$r\not\models\eexis{b?}\univ{\delta!}\false$, which implies that there is no formula in $\Lcoco$ that is 
logically equivalent to $\eexis{b?}\univ{\delta!}\false$.\qedhere
\end{itemize}
\end{proof}

\comen{
Finally we conclude this section comparing \iocos\ with a conformance simulation

\begin{definition}[conformance simulation~\cite{FabregasFP10}]
We say that a binary relation $R$ of states in a labelled transition system is a conformance simulation relation if and only if for any $(p,q)\in R$ the following conditions hold:
\begin{itemize}
\item For all $a\in A$, if $p\tran{a}$, then $q\tran{a}$.

\item For all $a\in A$ such that $q\tran{a}q'$ and $p\tran{a}$, there exists some $p'$ with $p\tran{a}p'$ and $(p',q')\in R$.
\end{itemize}
We will write $p\conf q$ if there exists a conformance simulation $R$ such that $(p,q)\in R$.
\end{definition}

\begin{definition}[conformance logic~\cite{FabregasFP10}]
The syntax of the logic for conformance simulation, denoted by \Lconf, is defined by the following {BNF} grammar.

$$\phi ::= \true \mid \false \mid \phi \wedge \phi
\mid \phi \vee \phi \mid \uuniv{a}\phi. $$

Where $a\in A$.
\end{definition}

\begin{proposition}
  Compare $\iocos$ with conformance~\cite{FabregasFP10} 
\end{proposition}
}

The comparison between the logics characterising the relations
\iocos\ and Tretmans' \ioco is of particular interest. We will delay
this comparison till Section~\ref{sec:relation_ioco}, as we will make
use of some of the results in the following sections.

\section{Adding fixed points to the logics}\label{sec:fix-points}

Even if an LTS is image-finite and has a finite number of states, it can capture infinite behaviours, for example by 
means of loops.  As is well known, the logics presented in Section~\ref{sec:basic} can only describe properties of finite 
fragments of process computations and can be made more expressive by extending them with fixed-point operators.

Here we only explicitly show how to add the greatest fixed-point
operator to one of the characterising logics for \iocos, $\tLiocos$,
presented in Section~\ref{sec:basic}. We will make use of this
extended logic in Section~\ref{sec:charac_formula} to define the
characteristic formula for processes with respect to \iocos.

In this section we follow the approach of the so-called equational $\mu$-calculus (see, for example, 
\cite{AILS12,Larsen90} to which we refer the reader for further details). We assume a countably infinite set 
$\mathsf{Var}$ of formula variables. The meaning of formula variables is specified by means of a declaration. A 
declaration is a function $D:\mathsf{Var} \rightarrow \tLiocos$ that associates a formula $D(X)$ with each variable 
$X$. Intuitively, if $D(X) = \phi$, then $X$ stands for the largest solution of the equation $X = \phi$. In what follows, 
we are only interested in the restriction of a declaration to a finite collection of formula variables, and we express such 
a declaration as a system of equations $X_1=\phi_1,\ldots, X_n=\phi_n$, with $n\geq 1$.

We interpret the language over the set $S$ of processes in some
LTS. Since a formula $\varphi$ may contain formula variables, its
semantics---that is, the set of processes in $S$ that satisfy
$\varphi$---is defined relative to an environment
$\sigma:\mathsf{Var}\rightarrow \mathcal{P}(S)$. Intuitively, $\sigma$
assigns to each variable the set of processes in $S$ for which the
variable holds true.

As is well known, the set of all environments $[\mathsf{Var}\rightarrow
    \mathcal{P}(S)]$ is a complete lattice with respect to the partial
  order induced by pointwise set inclusion.

\begin{definition}\label{Def;fixed-point-logic}
	The syntax of the logic for \iocos with greatest fixed points,
        denoted by \nLiocos, is defined by the following BNF grammar
	$$\phi ::= X\mid \true \mid \false \mid \phi \wedge \phi
	\mid \phi \vee \phi \mid \uuniv{a?} \phi \mid
	\univ{a!}\phi, $$
	
\noindent where $a?\in I$, $a!\in O$ and $X\in\mathsf{Var}$.
\end{definition}

The semantics of \nLiocos\ is given by:

\begin{itemize}
	\item $(\sigma,p)\models X$ iff $p\in\sigma(X)$.
	
	\item $(\sigma,p)\models \true$ for all $p$.
	
	\item  $(\sigma,p)\models \varphi_1\wedge\varphi_2$ iff $(\sigma,p)\models \varphi_1$ and $(\sigma,p)\models \varphi_2$.
	
	\item  $(\sigma,p)\models \varphi_1\vee\varphi_2$ iff $(\sigma,p)\models \varphi_1$ or $(\sigma,p)\models \varphi_2$.
	
	\item $(\sigma,p)\models \univ{a!} \phi$ iff $(\sigma,p')\models \phi$ for each  $p \tran{a!} p'$.
	
	\item $(\sigma,p)\models \uuniv{a?} \phi$ iff $p\tran{a?}$ and $(\sigma,p')\models \phi$,
	for each  $p \tran{a?} p'$.	
\end{itemize}
The semantics of each formula $\varphi$ is therefore the function
$\sem{\varphi}: [\mathsf{Var}\rightarrow\mathcal{P}(S)]\rightarrow
\mathcal{P}(S)$ defined thus:
\[
\sem{\varphi}\sigma = \{ p \mid (\sigma,p)\models \varphi \}. 
\]
A declaration function $D$ induces an endofunction $\sem{D}$ over the complete lattice $[\mathsf{Var}\rightarrow\mathcal{P}(S)]$ defined by
\[
(\sem{D}\sigma)(X)= \sem{D(X)}\sigma . 
\]

\begin{proposition}
  The function $\sem{D}$ is monotone. Indeed, so is $\sem{\varphi}$
  for each formula $\varphi$.
\end{proposition}

\begin{proof}
Since $\uuniv{a?}\varphi$ is logically equivalent to the HML formula $\exis{a?}\true\wedge\univ{a?}\varphi$, the result 
follows straightforwardly from the monotonicity of the semantic counterparts of the operators in HML (see 
\cite{Larsen90}, for example).
\end{proof}

Since $\sem{D}$ is a monotone endofunction over a complete lattice, it
has a greatest (and a least) fixed point $\sem{D}_{\max}$ by the
Knaster-Tarski Theorem~\cite{Tarski55}.


\begin{theorem}
	For all states $i, s$ in some LTS,
	$$i\iocos s \quad\mbox{iff}\quad \forall\phi\in\nLiocos\quad {s\models\phi}\Rightarrow{i\models\phi}.$$
\end{theorem}

\begin{proof}
Since $\tLiocos$ is a sublogic of $\nLiocos$, the \emph{if} 
implication of the proof follows by
Theorem~\ref{theo:tlogic_char}. For the \emph{only if} 
implication,
the proof follows by applying the standard techniques for \mucal\ and
bisimulation equivalence (see, for example, \cite[Section~$5.4$]{Stirling01}).
\end{proof}

\begin{remark}\label{rem:numu}
  We showed in Section~\ref{sec:basic} how $\Liocos$ and $\tLiocos$ are dual logics (Corollary~\ref{cor:tlogicANDlogic}). Analogously to the construction of $\nLiocos$ from $\tLiocos$, we could add $\mu$, the least fixed point operator, to the logic $\Liocos$, for defining $\mLiocos$.
\end{remark}

Logics like \nLiocos\ can be used in model checking in order to
evaluate if an implementation satisfies some desired properties. There
is already a tool support for \iocos\ implemented in the \mCRL tool
set \cite{mCRL2,GLM15,GrooteMousavi}, including a minimisation algorithm. This
allows one to model check \nLiocos\ using the \mCRL tool. In fact, the
built-in logic used in \mCRL is the \emph{first-order modal
  $\mu$-calculus}~\cite{Keiren13}, which is an extension of the
$\mu$-calculus where the actions are allowed to include data
parameters. We omit that extension here, but we show some examples of
properties that can be expressed with our logic.

\begin{example} 
Let us illustrate some useful properties that can be expressed using the logic $\nLiocos$.

\begin{itemize}
	\item In the theory of \ioco, implementations are supposed to
          be input enabled, that is, each of their reachable states
          should be able to perform every input action. This
          requirement can be expressed as follows:
	
	\[X= \bigwedge_{a?\in I}\uuniv{a?}X \wedge \bigwedge_{b!\in O}\univ{b!}X.
	\]
	
	\item The next formula formalises that a property $\varphi$
          must hold invariantly over input-enabled systems:
	
	\[X=\varphi \wedge (\bigwedge_{a?\in I} \uuniv{a?}X \wedge \bigwedge_{o!\in O}\univ{o!}X).
	\]
    
     That is, $\varphi$ must be true, and must remain true after performing any input or output action.
\end{itemize}

\end{example}

\section{Characteristic formulae}
\label{sec:charac_formula}

Once we have logically characterised {\iocos}, it is natural to search
for a single logical formula that characterises completely the
behaviour of a process up to that preorder; this is what is called a
\emph{characteristic formula}~\cite{GrafS84,SteffenI94}. In this
section we follow~\cite{AILS12} in order to define such a formula for
processes with respect to {\iocos}. We will consider the logic from
Section~\ref{sec:fix-points} with its greatest-fixed-point
interpretation given therein.

First, let us define the notion of characteristic formula for the particular case of the \iocos-semantics.

\begin{definition}\label{df:characteristic_formula}
	A formula $\chi_s$ is \emph{characteristic} for $s$ (with
        respect to {\iocos}) iff for all $i$ it holds that $i\models
        \chi_s$ if, and only if, $i \iocos s$. (Note that this also
        implies that $s\models \chi_s$.)
\end{definition}

The next proposition states the explicit definition of characteristic formulae in the {\iocos} framework.

\begin{proposition}\label{pro:characteristic_formula}
	The characteristic formula for a process $q$ in a finite LTS can be obtained recursively as:
	
	\begin{equation}\label{eq:char_formula}
	\chi_q=\bigwedge_{a?\in \ins(q)} \uuniv{a?}\bigvee_{q\tran{a?}q'}\chi_{q'} \wedge \bigwedge_{a!\in O} 
	 \univ{a!}\bigvee_{q\tran{a!}q'}\chi_{q'},
	\end{equation}
	that is, letting $D$ be the declaration defined by the equations of the form~(\ref{eq:char_formula}), we 
	have that: $(\sem{D}_{\max},p)\models \chi_q$ iff $p\iocos q$.
\end{proposition}

\begin{remark}
Note that $\chi_q$ constrains the behaviour of a process that
satisfies it only for those input actions that $q$ can perform,
exactly in the same way as the \iocos\ relation does.
\end{remark}

In order to prove Proposition~\ref{pro:characteristic_formula}, we
follow the general theory described in \cite{AILS12}. The concrete
goal is to apply Corollary 3.4 in \cite{AILS12}. As a first step, we
define \iocos as the greatest fixed-point of a monotone function,
denoted by $\Fiocos$, over the set of binary relations over $S$, the
set of states of an LTS.

\begin{definition}
For each $R\subseteq S\times S$, we have that $(p,q)\in\Fiocos(R)$ iff
\begin{enumerate}
	\item $\ins(q)\subseteq\ins(p)$, 
	\item for all $a? \in \ins(q)$,	if $p\tran{a?} p'$ then there exists some $q'$ such that $q\tran{a?} q'$ and $(p', q')\in R$, and
	\item for all $o!\in O$, if $p\tran{o!} p'$ then there exists some $q'$ such that $q\tran{o!} q'$ and $(p', q')\in R$.
\end{enumerate}
\end{definition}

Next, we show a result connecting $\Fiocos$ and $\chi$ defined in 
Proposition~\ref{pro:characteristic_formula}. 

\begin{lemma}\label{lem:monotone}
Let $R \subseteq S \times S$ and $p,q\in S$. Consider the set of variables $\mathsf{Var}=\{X_p\mid p\in S\}$ and let $\sigma_R$ be the interpretation 
defined as $\sigma_R(X_q)=\{p\mid (p,q)\in R\}$. Then,  $(p,q)\in\Fiocos(R)$ iff
	\[
	(\sigma_R,p)\models \bigwedge_{a?\in \ins(q)} \uuniv{a?}\bigvee_{q\tran{a?}q'} X_{q'} \wedge \bigwedge_{a!\in O} \univ{a!}\bigvee_{q\tran{a!}q'} X_{q'},
	\]
\end{lemma}

\begin{proof}
First, suppose that $(p,q)\in\Fiocos(R)$. We will show that 

\[(\sigma_R,p)\models \bigwedge_{a?\in \ins(q)} \uuniv{a?}\bigvee_{q\tran{a?}q'} X_{q'} \wedge \bigwedge_{a!\in O} \univ{a!}\bigvee_{q\tran{a!}q'} X_{q'}.
\]

To this end, let $a?\in \ins(q)$. We show first that $(\sigma_R,p)\models 
\uuniv{a?}\bigvee_{q\tran{a?}q'} X_{q'}$. Indeed, since $(p,q)\in\Fiocos(R)$, $\ins(q)\subseteq\ins(p)$. Hence, 
$a?\in \ins(p)$ and therefore $p\tran{a?}$. Assume now that $p\tran{a?}p'$, for some $p'$, we claim that 
$(\sigma_R,p')\models \bigvee_{q\tran{a?}q'} X_{q'}$. Indeed, since $(p,q)\in\Fiocos(R)$, $a?\in \ins(q)$ and 
$p\tran{a?}p'$, there exists some $q'$ such that $q\tran{a?}q'$ with $(p',q')\in R$, that is, such that $p'\in 
\sigma_R(X_{q'})$. Thus, we have that $(\sigma_R,p)\models \uuniv{a?}\bigvee_{q\tran{a?}q'} X_{q'}$.

Next, let us see that $(\sigma_R,p)\models \bigwedge_{a!\in O} \univ{a!}\bigvee_{q\tran{a!}q'} X_{q'}$. 
Analogously to the previous case let us assume that $p\tran{a!}p'$, for some $p'$. We claim that 
$(\sigma_R,p')\models \bigvee_{q\tran{a?}q'} X_{q'}$. Indeed, since $(p,q)\in\Fiocos(R)$ and $p\tran{a!}p'$, 
there exists some $q'$ such that $q\tran{a!}q'$ with $(p',q')\in R$, that is, such that $p'\in \sigma_R(X_{q'})$. 
Thus, we have that $(\sigma_R,p)\models \univ{a!}\bigvee_{q\tran{a?}q'} X_{q'}$. This completes the proof of 
the first implication of the theorem.

Now, let us suppose that 
\[(\sigma_R,p)\models \bigwedge_{a?\in \ins(q)} \uuniv{a?}\bigvee_{q\tran{a?}q'} X_{q'} \wedge \bigwedge_{a!\in O} \univ{a!}\bigvee_{q\tran{a!}q'} X_{q'} . \] 
We prove that $(p,q)\in\Fiocos(R)$. First, since for each $a?\in\ins(q)$, $(\sigma_R,p)\models \uuniv{a?}\bigvee_{q\tran{a?}q'} X_{q'}$, we obtain that $\ins(q)\subseteq\ins(p)$.

Next, let us consider some $a?\in\ins(q)$. Since $(\sigma_R,p)\models \uuniv{a?}\bigvee_{q\tran{a?}q'} X_{q'}$, we have that $(\sigma_R,p')\models \bigvee_{q\tran{a?}q'} X_{q'}$ for every $p\tran{a?}p'$. Hence, for every $p\tran{a?}p'$, there exists some $q\tran{a?}q'$ such that $(\sigma_R,p')\models X_{q'}$, that is, for every $p\tran{a?}p'$, there exists some $q\tran{a?}q'$ such that $(p',q')\in R$.

Finally, let us consider $a!\in O$, since $(\sigma_R,p)\models \univ{a!}\bigvee_{q\tran{a!}q'} X_{q'}$, we also have that $(\sigma_R,p')\models\bigvee_{q\tran{a!}q'} X_{q'}$ for every $p\tran{a!}p'$. Hence, for each $p\tran{a!}p'$ there exists $q\tran{a!}q'$ such that $(\sigma_R,p')\models X_{q'}$, that is, such that $(p',q')\in R$. This completes the proof of the lemma.
\end{proof}

In the light of Lemma~\ref{lem:monotone}, we can apply Corollary 3.4 in~\cite{AILS12} to prove  Proposition~\ref{pro:characteristic_formula}, which gives the explicit definition of the 
characteristic formula for \iocos.

\section{The relation between \iocos and \ioco}
\label{sec:relation_ioco}

Input-output conformance (\ioco) was introduced by Tretmans
in~\cite{Tre96}. The intuition behind \ioco is that a process $i$ is a
correct implementation of a specification $s$ if, for each sequence of
actions $\sigma$ allowed by the specification, all the possible
outputs from $i$ after having performed $\sigma$ are allowed by the
specification. This is formalised below in a setting in which all
actions are observable.

\begin{definition}
Let $(S,I,O,\tran{})$ be an LTS with inputs and outputs. We define the
traces of a state $p\in S$ as $\traces(p) = \{ \sigma\mid\exists
p'.~p\tran{\sigma}p'\}$. Given a trace $\sigma$, we define
$p~\after~\sigma=\{p'\mid p'\in S, ~p\tran{\sigma}p'\}$. For each $T
\subseteq S$, we set $\Out(T)=\bigcup_{p\in T}\outs(p)$. Finally, the
relation $\ioco\subseteq S\times S$ is defined as:
$$i\ioco s \textrm{ iff } \Out(i\,\after\,
\sigma)\subseteq\Out(s\,\after\, \sigma), \textrm{ for all
}\sigma\in\traces(s).$$
\end{definition}
As shown in~\cite[Theorem~1]{GLM13}, {\iocos} is included in {\ioco}.

In the setting of Tretmans' standard {\ioco} theory~\cite{Tre96}, only input-enabled implementations are 
considered. We recall that a state $i$ in an LTS is input enabled if every state $i'$ that is reachable from $i$ 
is able to perform every input action, that is, $i'\tran{a?}$ holds for each $a?\in I$.

Theorem~$2$ in~\cite{LM13} states that if $i$ is input enabled, $i \ioco s$ implies $i\iocos s$. This means 
that, when restricted to input-enabled implementations, {\ioco} and {\iocos} coincide, and therefore the 
logics characterising {\iocos} presented in this paper also characterise {\ioco} over that class of LTSs. 
Unfortunately, however, Theorem~$2$ in~\cite{LM13} does {\em not} hold, as shown in the following example.

\begin{example}\label{ex:ioco}	
	Let $s$ and $i$ be defined as follows, where we assume that $I=\{a?,b?\}$. 
	
	
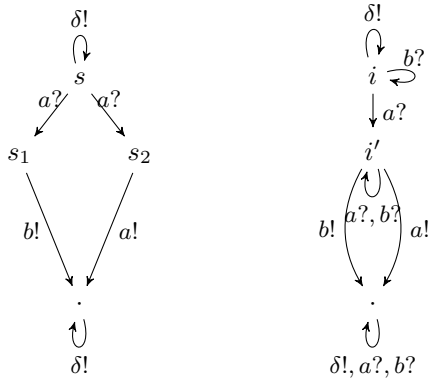
{\centering
		\begin{tikzpicture}[->,>=stealth',shorten >=1pt,auto]
		\matrix [matrix of math nodes, column sep={.8cm,between origins},row sep={1cm,between origins}]
		{
			& \node (C0) {s}; &
			&[1.5cm] &  \node (D0) {i};
			\\
			\node (C11) {s_1}; & & \node (C12) {s_2}; 
			&[1.5cm] & \node (D1) {i'}; \\[1cm]
			& \node (C21) {\cdot}; & &[1.5cm]
			& \node (D2) {\cdot};  \\
		};
		\begin{scope}[every node/.style={font=\small\itshape}]
		\path   
		(C0)   edge        		   node [left,above]  {$a?$} (C11)
		(C0)   edge        		   node [right,above]  {$a?$} (C12)
		(C0)   edge[loop above]    node [right,above] {$\delta!$} (C0)
		(C11)  edge                node [left]        {$b!$} (C21)
		(C12)  edge                node [right]        {$a!$} (C21)
		(C21)  edge[loop below]    node [right,below] {$\delta!$} (C21)
		(D0)   edge[loop above]    node [right,above] {$\delta!$} (D0)
		(D0)   edge[loop right]    node [right,above] {$b?$} (D0)
		(D0)   edge        		   node [right]        {$a?$} (D1)
		(D1)   edge[loop below]    node [below]        {$a?,b?$} (D1)
		(D1)   edge[bend right]    node [left]        {$b!$} (D2)
		(D1)   edge[bend left]     node [right]        {$a!$} (D2)
		(D2)   edge[loop below]    node [below]        {$\delta!,a?,b?$} (D2);
		\end{scope}
		\end{tikzpicture}
		
	}
	
	Note that $i$ is input-enabled, as required by the theory of \ioco. It
	is easy to see that $i\ioco s$. On the other hand, $i\niocos s$
	because each \iocos\ relation containing the pair $(i,s)$ would also
	have to contain the pair $(i',s_1)$ or the pair $(i',s_2)$. However,
	no relation including either of those pairs is an \iocos-relation
	because $i'\tran{a!}$ and $i'\tran{b!}$, but $s_1\notran{a!}$ and
	$s_2\notran{b!}$.
\end{example}

However, as one might expect, Theorem~$2$ in~\cite{LM13} becomes true
when the specification $s$ is deterministic. An LTS is deterministic
whenever there is only one possible successor for each action, that
is, if for all $p, p', p''\in S$ and $a\in L$, if $p\tran{a}p'$ and
$p\tran{a}p''$, then $p'=p''$.

\begin{proposition}
	Let $i$ be an input-enabled LTS and $s$ be a deterministic LTS. If $i \ioco s$ then $i \iocos s$.
\end{proposition}

\begin{proof}
  Assume that $i \ioco s$. We will show that the relation
  $$R=\{(p,q) \mid \exists \sigma\in \traces(s)\text{ such that } p\in i\after\sigma \text{ and } q\in s\after\sigma\}$$
  is a \iocos-relation. First, we observe that $(i,s)\in R$ since $\varepsilon\in\traces(s)$. Now, we show that any pair $(p,q)\in R$ satisfies the conditions of Definition~\ref{dfn:iocos}.
	\begin{itemize}
		\item Trivially, $\ins(q)\subseteq \ins(p)$. Indeed, since $i$ is input-enabled, $\ins(q)\subseteq I=\ins(p)$.
		
		\item  Let $a?\in\ins(q)$ and assume that $p\tran{a?}p'$ for some $p'$. Since $a?\in\ins(q)$, by 
		definition there exists some $q'$ such that $q\tran{a?}q'$, so $\sigma a?\in \traces(s)$. Hence, $p'\in 
		i\after\sigma a?$ and $q'\in s\after\sigma a?$, that is, $(p',q')\in R$.
		
		\item  Let $b!\in O$ be such that $p\tran{b!}p'$, that is, $b!\in \outs(p)\subseteq \Out(i~{\after}~\sigma)$. 
		Now, since $i~{\ioco}~s$ and $\sigma\in \traces(s)$, we have that $\Out(i~{\after}~\sigma)\subseteq\Out(s~{\after}~\sigma)$. Hence, $b!\in\Out(s\after \sigma)$ and, since $s$ is deterministic, 
		$b!\in\outs(q)$ so there exists some $q'$ such that $q\tran{b!}q'$. Summing up, we have that $\sigma b!\in \traces(s)$, 
		$p'\in i~{\after}~\sigma b!$ and $q'\in s~{\after}~\sigma b!$, that is, $(p',q')\in R$. \qedhere
	\end{itemize}
	
\end{proof}

\begin{remark}
The construction of the relation $R$ used in the proof of the above
proposition is akin to the definition of `coinductive {\ioco}' employed in
the proof of Theorem 3 in~\cite{NorooziMW13}. In fact, the use of such
relations in the algorithmics of decorated  trace semantics, such as
{\ioco}, as well as failure and testing equivalences, can be traced at
least as far back as~\cite{CleavelandH93}.
\end{remark}

\subsection{The relation with a logic for \ioco}\label{subsec:rel_ioco}

In~\cite{BoharMousavi14} Beohar and Mousavi introduced an explicit
logical characterisation of \ioco. This characterisation uses a
non-standard modal operator reminiscent of our $\uuniv{\cdot}$,
denoted by $\uioco{\cdot}$\footnote{In fact, the symbol used to denote
  the operator $\uioco{\cdot}$ in~\cite{BoharMousavi14} is $\langle\![
    \cdot ]\!\rangle$, but we prefer to use an alternative notation in order
  to avoid confusion with our modal operator $\eexis{\cdot}$.}. However, output actions can also be used as labels of $\uioco{\cdot}$. 
This modality can be extended to traces $\sigma$ as follows: $p \models
\uioco{\sigma} \phi$ if, and only if, $p\tran{\sigma}$ and $p' \models
\phi$, for each $p'$ such that $p \tran{\sigma} p'$. (Note that, for the
particular case of input actions $a?$, the semantics of $\uioco{a?}$
coincides with that of $\uuniv{a?}$.)

The explicit logical characterisation of {\ioco} given
in~\cite{BoharMousavi14} is defined by means of two
different subclasses of logical formulae. The first subclass permits
only formulae of the form $\uioco{\sigma}\univ{b}\false$, where $\sigma$
is a trace and $b$ is an output action.

For the second subclass of formulae, Beohar and Mousavi consider the natural
extension of the operator $\univ{\cdot}$ to traces, defined as: $p
\models \univ{\sigma} \phi$ if, and only if, $p' \models \phi$ for each
$p'$ such that $p \tran{\sigma} p'$. This second subclass permits only
formulae of the form $\univ{\sigma}\univ{b}\false$, where $\sigma$ is a
trace and $b$ is an output action.

The formulae in each of these two subclasses characterise one defining
property of the \ioco-relation. This intuition is made precise in the
following lemma. 

\begin{lemma}[\cite{BoharMousavi14}]\label{lem:ioco2logic}
	For each sequence of actions $\sigma$, output action $b$ and
        process $p$ the following statements hold:
	\begin{enumerate}
		\item\label{key1} $\sigma\in\traces(p)$ and $b\notin\Out(p ~\after~\sigma)$ iff $p\models\uioco{\sigma}\univ{b}\false$.
		
		\item $b\notin\Out(p~\after~\sigma)$ iff $p\models\univ{\sigma}\univ{b}\false$.
	\end{enumerate}
\end{lemma} 
The resulting logical characterisation theorem by Beohar and Mousavi for \ioco\ is as follows.

\begin{theorem}[\cite{BoharMousavi14}]\label{thm:BMioco}
	$i\ioco s$ iff, for all $\sigma\in L^*$, $b\in O$, if $s\models\uioco{\sigma}\univ{b}\false$, then $i\models\univ{\sigma}\univ{b}\false$.
\end{theorem}

Theorem~\ref{theo:tlogic_char} in this paper is the counterpart of the above result for \iocos. Note, however, 
that Theorem~\ref{thm:BMioco} is not a classic modal characterisation result (as it is the case of, for 
example, Theorem~\ref{theo:tlogic_char}) where if the implementation $i$ is correct with respect to the 
specification $s$ and $s$ satisfies a formula, then also $i$ satisfies it. Here the implementation does not 
need to satisfy the properties that hold for the specification. By way of example, implementations need not 
exhibit all the traces of a specification they correctly implement.

As we will now argue, the logics for \ioco\ and \iocos\ are incomparable in terms of their expressive power. 
First of all, note that, if we consider only input-enabled implementations, the formulae of the form 
$\univ{\sigma}\univ{b}\false$, with $\sigma$ a trace, can be expressed in $\tLiocos$ since in an 
input-enabled scenario $\uuniv{a?}$ has the same semantics as $\univ{a?}$. On the other hand, it is not 
possible to define a formula $\phi\in\tLiocos$ that captures Lemma~\ref{lem:ioco2logic}(\ref{key1}). Indeed, 
by way of example, consider $\phi=\uioco{x!}\univ{b!}\false$. Any specification $s$ would have to satisfy 
$\phi$ iff $s\tran{x!}$ and $s'\models\univ{b!}\false$, for all $s\tran{x!}s'$. Now, assume that we have in 
\tLiocos\ a formula $\psi$ whose semantics coincides with that of $\uioco{x!}\univ{b!}\false$. Let

{\centering
	\begin{tikzpicture}[->,>=stealth',shorten >=1pt,auto]
	\matrix [matrix of math nodes, column sep={4cm,between origins},row sep={1cm,between origins}]
	{
		\node (C0) {s}; & \node (D0) {i};\\
	};
	\begin{scope}[every node/.style={font=\small\itshape}]
	\path   
	(C0)   edge[loop right]	   node  {$x!$} (C0)
	(C0)   edge[loop left]	   node  {$a!$} (C0)
	(D0)   edge[loop left]     node  {$a!$} (D0);
	\end{scope}
	\end{tikzpicture}
	
}

\noindent 
It is easy to see that $i~\iocos~s$, but $s\models\psi$ and $i\not\models\psi$. In other words, $\psi$ is a formula that distinguishes processes related by \iocos. Hence, such a formula $\psi$ cannot be expressed in any logic that characterises \iocos.

On the other hand, let us consider the two processes of Example~\ref{ex:ioco} and the formula $\phi=\uuniv{a?}(\univ{a!}\false \vee \univ{b!}\false)\in\tLiocos$. As we already stated in Example~\ref{ex:ioco}, $i\ioco s$, but $s\models\phi$ and $i\not\models\phi$. Hence, $\phi$ can distinguish processes that are \ioco-related. 

Now, since \iocos implies \ioco~\cite{GLM13}, we can use $\nLiocos$ (see Section~\ref{sec:fix-points}) to express distinguishing formulae whenever two processes are not \ioco-related.

\begin{proposition}\label{p:distinguishig_formula}
  If $i\nioco s$, then there exists a formula $\psi\in\nLiocos$ such that $s\models\psi$ but $i\not\models\psi$. 
\end{proposition}

\begin{proof}
	Straightforward from the fact that $i\iocos s$ implies $i\ioco s$, the logical characterisations of \Liocos\ in Theorem~\ref{theo:logic_char} and \tLiocos\ in Theorem~\ref{theo:tlogic_char}.
\end{proof}

\begin{remark}
  As we commented on Remark~\ref{rem:numu}, a dual logic to $\nLiocos$ is
  $\mLiocos$. Therefore, analogously, if $i\nioco\ s$, then there exists a
  formula $\phi\in\mLiocos$ such that $i\models\phi$ but $s\not\models\phi$.
\end{remark}

Proposition~\ref{p:distinguishig_formula} can be applied, in particular, for $\chi_s$, the characteristic formula 
of $s$ (Definition~\ref{df:characteristic_formula}). Hence, if $i\nioco\ s$, we obtain that $s\models\chi_s$ 
and $i\not\models\chi_s$. That is, the characteristic formula of $s$ can be used as a distinguishing formula. 
This way we have an alternative, albeit incomplete, criterion for checking if an implementation does not 
conform to a specification with respect to \ioco:
\begin{inparaenum}[$(i)$]
	\item first, build $\chi_s$, the characteristic formula for $s$;
	\item check if $i\models\chi_s$ (\iocos satisfaction); 
	\item now, if $i\not\models\chi_s$ then, we have that $i\nioco s$.
\end{inparaenum}

\section{Rule formats for \iocos}\label{Sect:ruleformat}

In this section we focus on the study of structural properties of
\iocos, with emphasis on its compositionality and on 
compositional proof systems for its characterising logic {\Liocos}.
We start, in Section~\ref{ssec:congruence}, by defining a
precongruence rule format for \iocos. That rule format is shown to
ensure compositionality with respect to \iocos\ for a subset of the
operators defined using rules in the GSOS format proposed by Bloom,
Istrail and Meyer~\cite{BIM95}. The rule format provides a sufficient
condition for compositionality; we show with counterexamples how the
rules we propose cannot be easily relaxed without jeopardizing the
compositionality result for \iocos.

Next, in Section~\ref{ssec:decomposition}, we use the logical characterisation of Section~\ref{sec:basic} and 
the general \emph{modal decomposition} methodology of Fokkink and van Glabbeek~\cite{FokkinkGW06}. 
For the operators in the rule format of Section~\ref{ssec:congruence}, we show how to check the satisfaction 
of a logical formula in $\Liocos$ in a compositional way.

Finally, quiescent behaviour is an important issue in the theory of \ioco/\iocos. In 
Section~\ref{ssec:rule-quiescent} we develop a rule format to ensure coherent quiescent behaviour in the 
sense of Definition~\ref{dfn:lts}.

The restriction in this section to GSOS rules is partly justified by our wish to have a purely syntactic rule format and by the undecidability results presented in~\cite{KlinNachyla16}. In what follows, we assume that the reader is familiar with the standard notions of signature and terms over a signature.

\subsection{Congruence rule format}\label{ssec:congruence}
We recall that a deduction rule for an operator $f$ of arity $n$ in
some signature $\Sigma$ is in the {\em GSOS format} if, and only if, it
has the following form:
\begin{equation}\label{eqn:gsos-rule}
\sosrule{\{x_{i} \tran{a_{ij}} y_{ij}  \mid 1 \leq i \leq n, 1 \leq j \leq m_i\} \cup \{x_{i} \notran{b_{ik}}  \mid 1 \leq i \leq n, 1 \leq k \leq \ell_i \}}
{f(\arr{x}) \tran{a} C[\vec{x}, \vec{y}]}
\end{equation}
where the $x_i$'s and the $y_{ij}$'s $(1 \leq i \leq n$ and $1 \leq j
\leq m_i)$ are all distinct variables, $m_{i}$ and $\ell_{i}$ are natural
numbers, $C[\vec{x}, \vec{y}]$ is a term over $\Sigma$ with variables
including at most the $x_{i}$'s and $y_{ij}$'s, and the $a_{ij}$'s,
$b_{ik}$'s and $a$ are actions from $L$. The above rule is said to be
{\em $f$-defining} and {\em $a$-emitting}. Its {\em positive trigger
  for variable $x_i$} is the set $\{a_{ij} \mid 1 \leq j \leq m_i\}$ and its {\em negative trigger for
  variable $x_i$} is the set $\{b_{ik} \mid 1 \leq k \leq
\ell_i\}$. The {\em source in the conclusion} of the rule is
$f(\arr{x})$.

A GSOS language is a triple $(\Sigma, L, D)$ where $\Sigma$ is a
finite signature, $L$ is a finite set of labels and $D$ is a finite set of
deduction rules in the GSOS format. In what follows, we assume,
without loss of generality, that all $f$-defining rules have the same
source of their conclusions.

A GSOS language naturally defines a set of transitions over the
variable-free terms over $\Sigma$ by structural induction: for vectors
of such terms $\arr{p}$ (with typical entry $p_i$) and $\arr{q}$ (with
entries $q_{ij}$), there is a transition ${f(\arr{p}) \tran{a}
  C[\vec{p}, \vec{q}]}$ if, and only if, there is an $f$-defining rule
of the form (\ref{eqn:gsos-rule}) such that
\begin{itemize} 
\item $p_{i} \tran{a_{ij}} q_{ij}$ for each $1 \leq i \leq n$ and $1 \leq j \leq m_i$ and 
\item $p_{i} \notran{b_{ik}}$ for each $1 \leq i \leq
  n$ and $1 \leq k \leq \ell_i$. 
\end{itemize}
Note that GSOS rules define operations over states in an arbitrary LTS
with inputs and outputs.  In what follows, we apply derived operations
built over the signature of a GSOS language to states in the
collection of LTSs with input and output actions.

\begin{definition}\label{Def:ruleformat}
An operation $f$ in a GSOS language is in {\em \iocos-format} if the collection of $f$-defining rules satisfies the
following conditions:
\begin{enumerate}
\item \label{IC1} Each $a?$-emitting rule, where $a?$ is an input
  action, has only output actions as labels of negative premises and
  input actions as labels of positive premises.

\item \label{IC2} For each input action $a?$ and each pair of rules $r
  = \frac{H}{f(x_1,...,x_n) \tran{a?} t}$ and $r' =
    \frac{H'}{f(x_1,...,x_n) \tran{a?} t'}$, there is a rule $r'' =
      \frac{H''}{f(x_1,...,x_n) \tran{a?} t'}$ such that
        \begin{enumerate}
        \item \label{postrig} for each $1\leq i\leq n$, the positive
          trigger for variable $x_i$ in $r''$ is included in the
          positive trigger for variable $x_i$ in $r$; 
        \item \label{negtrig} for each $1\leq i\leq n$, the negative
          trigger for variable $x_i$ in $r''$ is included in the
          negative trigger for variable $x_i$ in $r$;   
        \item \label{posprem} if $x_i \tran{b?} z$ is contained in
          $H''$ and $z$ occurs in $t'$, then $x_i \tran{b?} z$ is also
          contained in $H'$.
        \end{enumerate}

\item \label{IC3} Each $a!$-emitting rule, where $a!$ is an output
  action, has only input actions as labels of negative premises and
  output actions as labels of positive premises.
\end{enumerate}
A GSOS language is in {\em \iocos-format} if so is each of
its operations.
\end{definition}

Next, we state the main result of this section.

\begin{theorem}\label{thm:congruence}
$\iocos$ is a precongruence for each GSOS language in {\iocos\ format}. 
\begin{proof}
The proof of this result may be found in \ref{App:thmcongruence}.  
\end{proof}
\end{theorem}

As an example of application of the above result, we show that the
merge operator from~\cite{BenesDHKN15} can be expressed in our
rule format.

\begin{example}\label{Ex:merge}
Merge, or conjunction, is a composition operator from the theory of
\ioco. It acts as a logical conjunction of requirements, that is, it
describes systems by a conjunction of sub-systems, or
sub-specifications. We denote by $\bigwedge_{i=1}^{n}s_i$ the result
of the merge of the states $s_i$, with $1\leq i\leq
n$. In~\cite{BenesDHKN15} it is noted that, in general, the merge of two
systems can lead to invalid states (for example the merge of a
quiescent state with another with some output). The solution is to add
a pruning algorithm after calculating the merge. Here we just show the
merge operator and not that pruning algorithm.

The merge operator can be formalised using the following GSOS rules
(one such rule for each $a\in L$):
 \begin{mathpar}
\inferrule*[right={\normalsize .}]
{\{x_i\tran{a}y_i\mid 1\leq i\leq n\}}
{\bigwedge_{i=1}^{n}x_i\tran{a}\bigwedge_{i=1}^{n}y_i}
\end{mathpar}

It is immediate to check that the above rules are in
\iocos-format. Therefore the above theorem yields that the merge
operator preserves {\iocos}.
\end{example}
In the following examples, we discuss whether the rules specifying
(variations on) some classic process-algebraic operations and those
considered in~\cite{BijlRT03} meet the constraints of our rule format.

\begin{example}[Nondeterministic choice]\label{Ex:choice}
The following rules describe the behaviour of the nondeterministic
choice operation considered in~\cite[Figure~1]{GLM18}, which is a minor
variation on the classic CCS choice~\cite{Mil89CC}:

 \begin{mathpar}
\inferrule*[]
{x_1\tran{a}y_1}
{x_1+x_2\tran{a}y_1} \qquad
\inferrule*[]
{x_2\tran{a}y_2}
{x_1+x_2\tran{a}y_2} \qquad
\inferrule*[right={\normalsize ,}]
{x_1\tran{\delta!}y_1,~x_2\tran{\delta!}y_2}
{x_1+x_2\tran{\delta!}y_1 + y_2}  
 \end{mathpar}
 where $a\in L\setminus \{\delta!\}$. (Note that the third rule above
 is an equivalent reformulation of rule VII given
 in~\cite[Figure~1]{GLM18}, which is not in GSOS format.)

 These rules are not in {\iocos-format}. To see this, consider the
 $a?$-emitting rules for some input action $a?$, namely
  \begin{mathpar}
r = \frac
{x_1\tran{a?}y_1}
{x_1+x_2\tran{a?}y_1}
\quad\text{ and }\quad
r' = \frac
{x_2\tran{a?}y_2}
{x_1+x_2\tran{a?}y_2} . 
  \end{mathpar}
Those rules do not satisfy condition~\ref{IC2} in
Definition~\ref{Def:ruleformat}. Indeed, the only possible choice for
rule $r''$ is $r'$, which satisfies neither requirement~\ref{postrig}
nor requirement~\ref{posprem}.

It turns out that nondeterministic choice does not preserve {\iocos}
and the reader will find it easy to construct an example witnessing
the fact that {\iocos} is not a precongruence with respect to $+$ by
considering processes that do not satisfy the proviso
of~\cite[Proposition~6.2]{GLM18}. 
\end{example}

\begin{example}[Interleaving]\label{Ex:par}
The binary version of the merge operator from Example~\ref{Ex:merge}
is essentially the synchronous parallel composition of its
arguments. Another standard notion of parallel composition operator is
one that interleaves the computational steps of its
arguments~\cite{Hoa85CSP}. Ignoring quiescence, the following rules
describe the behaviour of the interleaving operator:

 \begin{mathpar}
\inferrule*[]
{x_1\tran{a}y_1}
{x_1 ~|~ x_2\tran{a}y_1 ~|~ x_2} \qquad
\inferrule*[]
{x_2\tran{a}y_2}
{x_1 ~|~ x_2\tran{a}x_1 ~|~ y_2} 
 \end{mathpar}
 where $a\in L\setminus \{\delta!\}$. 

 These rules are not in {\iocos-format}. To see this, consider the
 $a?$-emitting rules for some input action $a?$, namely
  \begin{mathpar}
r = \frac
{x_1\tran{a?}y_1}
{x_1 ~|~ x_2\tran{a?}y_1 ~|~ x_2}
\quad\text{ and }\quad
r' = \frac
{x_2\tran{a?}y_2}
{x_1 ~|~ x_2\tran{a?}x_1 ~|~ y_2} . 
  \end{mathpar}
Those rules do not satisfy condition~\ref{IC2} in
Definition~\ref{Def:ruleformat}. Indeed, the only possible choice for
rule $r''$ is $r'$, which satisfies neither requirement~\ref{postrig}
nor requirement~\ref{posprem}.

Again, it turns out that the interleaving operation does not preserve
{\iocos} and the reader will find it easy to construct an example
witnessing the fact that {\iocos} is not a precongruence with respect
to $|$. The same applies to the unrestricted version of the parallel
composition operator considered in~\cite[Definition~2.3]{BijlRT03} in
the setting of {\ioco}. Note that the notion of parallel composition
considered in that reference is a partial operation, as it is only
defined for arguments whose sets of input actions and whose sets of
output actions are disjoint. 
\end{example}

\begin{example}[Relabelling]\label{Ex:relabelling}
In this example, we consider a variation on the CCS relabelling
operator~\cite{Mil89CC} that is appropriate in the setting of
labelled transition systems with inputs and outputs. A
\emph{relabelling} is a function $f: L \rightarrow L$ that maps input
actions to input actions, output actions to output actions and such
that $f(\delta!) = \delta!$. The relabelling operator $\_[f]$ associated with
a relabelling $f$ is specified using the following rules:

 \begin{mathpar}
\frac
{x\tran{a}y}
{x[f]\tran{f(a)}y[f]}  \qquad \text{where $a\in L$.}
 \end{mathpar}
If $f$ is injective over the set of input actions $I$,
then the above rules are in {\iocos} format and
Theorem~\ref{thm:congruence} yields that the operator $\_[f]$ preserves
{\iocos}.

On the other hand, assume that $f(a?) = f(b?) = a?$ for two distinct
input actions $a?$ and $b?$. Then, the collection of rules for $\_[f]$
includes the rules
  \begin{mathpar}
r = \frac
{x\tran{a?}y}
{x[f]\tran{a?}y[f]}
\quad\text{ and }\quad
r' = \frac
{x\tran{b?}y}
{x[f]\tran{a?}y[f]} . 
  \end{mathpar}
Those rules do not satisfy condition~\ref{IC2} in
Definition~\ref{Def:ruleformat}. Indeed, choosing $r'$ as $r''$
violates requirement~\ref{postrig} and choosing $r$ as $r''$ violates
requirement~\ref{posprem}. Therefore the rules for $\_[f]$ are not in
{\iocos-format}.  The reader will have no trouble in constructing an example
witnessing the fact that, for such a relabelling $f$, the operator
$\_[f]$ does not preserve {\iocos}.
\end{example}

\begin{example}[Restriction]\label{Ex:restriction}
The restriction 
operator $\_\setminus A$, where $A$ is included in $L\setminus \{\delta!\}$,  is specified using the following rules:

 \begin{mathpar}
\frac
{x\tran{a}y}
{x\setminus A\tran{a}y\setminus A}  \qquad \text{where $a\not\in A$.}
 \end{mathpar}
It is easy to see that the above rules are in {\iocos}-format and
Theorem~\ref{thm:congruence} yields that the restriction operator
$\_\setminus A$ preserves {\iocos}.
\end{example}
The main lesson one can draw from the examples we have presented is
that, similarly to {\ioco}, the {\iocos} relation is not algebraically
well behaved. Indeed, a variety of operators only preserve {\iocos} if
one makes some assumptions on the sets of input and output actions of
their arguments.

The operational specification of the hiding operator
$\mathbf{hide}~V~\mathbf{in}~\_$, where $V\subseteq O$, presented
in~\cite[Definition~2.3]{BijlRT03} requires an extension of the theory
of {\iocos} with the internal action $\tau\not\in L$. We briefly
discuss how this extension can be carried out and how to extend the
{\iocos}-format to cover the hiding operator in
Section~\ref{sec:conclusion}.

\subsubsection{The rule format cannot be relaxed easily}\label{ssec:counterexamples}

Here we present some examples showing that the restrictions of the rule format from Definition~\ref{Def:ruleformat} cannot be relaxed easily.

\begin{example}\label{ex:counter1}
	This example indicates that the use of input actions in negative
	premises of input-emitting rules would invalidate
	Theorem~\ref{thm:congruence}. Let $f$ be defined by the following
	rules:
	\begin{mathpar}
		\inferrule*[]
		{ }
		{f(x)\tran{\delta!}f(x)}
		\and
		\inferrule*[right={\normalsize ,}]
		{x\notran{a?}}
		{f(x)\tran{a?}f(x)}
	\end{mathpar}
	
	\noindent where $a?\in I$. Now, let us consider the following processes from Example~\ref{Ex:iocos-procs}:
	
	{\centering \begin{tikzpicture}[->,>=stealth',shorten
		>=1pt,auto] \matrix [matrix of math nodes, column
		sep={2cm,between origins},row sep={1cm,between origins}] {
			\node (C0) {i}; & \node (D0) {s};\\
		};
		\begin{scope}[every node/.style={font=\small\itshape}]
		\path   
		(C0)   edge[loop right]	   node  {$a?$} (C0)
		(C0)   edge[loop left]	   node  {$\delta!$} (C0)
		(D0)   edge[loop right]     node  {$\delta!$} (D0);
		\end{scope}
		\end{tikzpicture}
		
	}
	
	As remarked in Example~\ref{Ex:iocos-procs}, $i~{\iocos}~s$. On the other
	hand, $f(i)~{\niocos}~f(s)$ because $a?\in\ins(f(s))$ but
	$a?\notin\ins(f(i))$.
\end{example}

\begin{example}
	This example indicates that the use of output actions in positive
	premises of input-emitting rules would invalidate
	Theorem~\ref{thm:congruence}. Let $f$ be defined by the following
	rules:
	\begin{mathpar}
		\inferrule*[]
		{ }
		{f(x)\tran{\delta!}f(x)}
		\and
		\inferrule*[right={\normalsize ,}]
		{x\tran{b!}y}
		{f(x)\tran{a?}f(x)}
	\end{mathpar}
	
	\noindent where $b!\in O$. Now, let us consider the following processes:
	
	{\centering
		\begin{tikzpicture}[->,>=stealth',shorten >=1pt,auto]
		\matrix [matrix of math nodes, column sep={2cm,between origins},row sep={1cm,between origins}]
		{
			\node (C0) {p}; & \node (D0) {q};\\
		};
		\begin{scope}[every node/.style={font=\small\itshape}]
		\path   
		(C0)   edge[loop left]	   node  {$c!$} (C0)
		(D0)   edge[loop left]	   node  {$b!$} (D0)
		(D0)   edge[loop right]     node  {$c!$} (D0);
		\end{scope}
		\end{tikzpicture}
		
	}
	
	\noindent where $c!\in O$. We have that $p\iocos q$ but $f(p)\niocos f(q)$. Indeed, $a?\in\ins(f(q))$ but $a?\notin\ins(f(p))$.
\end{example}

\begin{example}\label{ex:counter-example}
	This example indicates that not meeting requirement~\ref{postrig} in Definition~\ref{Def:ruleformat} 	
	would invalidate Theorem~\ref{thm:congruence}. Let $f$ be defined by the following rules:
	\begin{mathpar}
		\inferrule*[]
		{ }
		{f(x)\tran{\delta!}f(x)}
		\and
		\inferrule*[]
		{ }
		{0\tran{\delta!}0}
		\and
		\inferrule*[]
		{x\tran{a?}y}
		{f(x)\tran{a?}0}
		\and
		\inferrule*[right={\normalsize .}]
		{ }
		{f(x)\tran{a?}f(x)}
	\end{mathpar}
	
	\noindent where $a?\in I$. The above set of rules does not meet requirement~\ref{postrig} in Definition~\ref{Def:ruleformat}. To see this, take 
	\begin{mathpar} 
		\inferrule*[right={\normalsize }] { }
		{f(x)\tran{a?}f(x)} 
	\end{mathpar}
	as rule $r$ and 
	\begin{mathpar} 
		\inferrule*[] {x\tran{a?}y}
		{f(x)\tran{a?}0} \end{mathpar} 
	as rule $r'$. Note that the only possible choice for
	rule $r''$ is $r'$ itself. However, with this choice, the positive
	trigger for variable $x$ in $r''$ is $\{a?\}$, which is not included
	in the positive trigger for variable $x$ in $r$, which is the empty
	set.
	
	Now, let us consider again the two processes of Example~\ref{ex:counter1}:
	
	{\centering
		\begin{tikzpicture}[->,>=stealth',shorten >=1pt,auto]
		\matrix [matrix of math nodes, column sep={2cm,between origins},row sep={1cm,between origins}]
		{
			\node (C0) {i}; & \node (D0) {s};\\
		};
		\begin{scope}[every node/.style={font=\small\itshape}]
		\path   
		(C0)   edge[loop left]	   node  {$\delta!$} (C0)
		(C0)   edge[loop right]	   node  {$a?$} (C0)
		(D0)   edge[loop right]     node  {$\delta!$} (D0);
		\end{scope}
		\end{tikzpicture}
		
	}
	
	\noindent We know that $i~\iocos~s$. On the other hand, we claim that $f(i)~\niocos~f(s)$. Indeed, $f(i)\tran{a?}0$ and the only way that $f(s)$ can match an $a?$-transition is by $f(s)\tran{a?}f(s)$. Since $0~\niocos~f(s)$ (because $a?\notin\ins(0)$), this implies that $f(i)\niocos f(s)$. 
	
	Note that requirements~\ref{negtrig} and~\ref{posprem} in
	Definition~\ref{Def:ruleformat} are instead met by taking $r'=r''$.
\end{example}

\begin{example}
	This example indicates that not meeting requirement~\ref{negtrig} in Definition~\ref{Def:ruleformat} 
	would invalidate Theorem~\ref{thm:congruence}.
	Let $f$ be defined by the
	following rules: \begin{mathpar} \inferrule*[] { }
		{f(x)\tran{\delta!}f(x)} \and \inferrule*[] { }
		{0\tran{\delta!}0} \and \inferrule*[] {x\notran{b!}}
		{f(x)\tran{a?}0} \and \inferrule*[right={\normalsize ,}] { }
		{f(x)\tran{a?}f(x)} \end{mathpar}
	\noindent where $a?\in I$ and $b!\in O$. The above set of rules does not meet requirement~\ref{negtrig} in Definition~\ref{Def:ruleformat}. To see this, take 
	\begin{mathpar} 
		\inferrule*[right={\normalsize }] { }
		{f(x)\tran{a?}f(x)} 
	\end{mathpar}
	as rule $r$ and 
	\begin{mathpar} 
		\inferrule*[] {x\notran{b!}}
		{f(x)\tran{a?}0} 
	\end{mathpar} 
	as rule $r'$. Note that the only possible choice for
	rule $r''$ is $r'$ itself. However, with this choice, the negative
	trigger for variable $x$ in $r''$ is $\{b!\}$, which is not included
	in the negative trigger for variable $x$ in $r$, which is the empty
	set.
	
	Now, let us consider the following processes:
	
	{\centering
		\begin{tikzpicture}[->,>=stealth',shorten >=1pt,auto]
		\matrix [matrix of math nodes, column sep={4cm,between origins},row sep={1cm,between origins}]
		{
			\node (C0) {p}; & \node (D0) {q};\\
		};
		\begin{scope}[every node/.style={font=\small\itshape}]
		\path   
		(C0)   edge[loop left]	   node  {$c!$} (C0)
		(C0)   edge[loop right]	   node  {$a?$} (C0)
		(D0)   edge[loop right]     node  {$c!$} (D0)
		(D0)   edge[loop left]     node  {$b!$} (D0);
		\end{scope}
		\end{tikzpicture}
		
	}
	
	\noindent where $c!\in O$. Again, $p\iocos q$, but $f(p)\niocos f(q)$. Indeed, $f(p)\tran{a?}0$ and the only way $f(q)$ can match an $a?$-transition is by $f(q)\tran{a?}f(q)$. However, as in the previous case, $0\niocos f(q)$.
	
	Note that requirements~\ref{postrig} and~\ref{posprem} in
	Definition~\ref{Def:ruleformat} are instead met by taking $r'=r''$.
\end{example}

\begin{example}
	This example indicates that not meeting requirement~\ref{posprem} in
	Definition~\ref{Def:ruleformat} would invalidate
	Theorem~\ref{thm:congruence}. Let $f$ be defined by the following rules:
	\begin{mathpar}
		\inferrule*[]
		{ }
		{f(x)\tran{\delta!}f(x)}
		\and
		\inferrule*[]
		{x \tran{a?} y}
		{f(x)\tran{a?} y}
		\and
		\inferrule*[]
		{x\tran{b?}y}
		{f(x)\tran{a?}y}
		\and
		\inferrule*[right={\normalsize ,}]
		{x \tran{a?} y }
		{f(x)\tran{a?}f(x)}
	\end{mathpar}
	
	\noindent where $a?,b?\in I$. The above set of rules does not meet requirement~\ref{posprem} in Definition~\ref{Def:ruleformat}. To see this, take 
	\begin{mathpar} 
		\inferrule*[right={\normalsize }] {x \tran{a?} y }
		{f(x)\tran{a?}f(x)} 
	\end{mathpar}
	as rule $r$ and 
	\begin{mathpar} 
		\inferrule*[] {x\tran{b?}y}
		{f(x)\tran{a?}y} \end{mathpar} as rule $r'$. Note that the only
	possible choice for rule $r''$ meeting requirement~\ref{postrig} and
	having $y$ as target of its conclusion is 
	\begin{mathpar}
		\inferrule*[right={\normalsize .}] {x\tran{a?}y}
		{f(x)\tran{a?}y} \end{mathpar}
	However, with this choice, the positive premise $x\tran{a?}y$ of $r''$ is not a positive premise of $r'$. 
	
	Now, let us consider the following two processes $p$ and $q$:
	
	{\centering
		\begin{tikzpicture}[->,>=stealth',shorten >=1pt,auto]
		\matrix [matrix of math nodes, column sep={1.5cm,between origins},row sep={1cm,between origins}]
		{
			\node (C0) {p}; & \node (C1) {0}; &[.5cm] \node (D0) {q};\\
		};
		\begin{scope}[every node/.style={font=\small\itshape}]
		\path   
		(C0)   edge[loop left]	  node  {$a?$,$\delta!$} (C0)
		(C0)   edge         	  node[above]  {$b?$} (C1)
		(C1)   edge[loop right]    node  {$\delta!$} (C1)
		(D0)   edge[loop right]    node  {$a?$,$\delta!$} (D0);
		\end{scope}
		\end{tikzpicture}
		
	}
	
	\noindent It is easy to see that $p~\iocos~q$. However, $f(p)~\niocos~f(q)$. To see this, observe that $f(q)$ can perform the input action $a?$ and $f(p) \tran{a?}0$ by rule $r'$. The only two possible matching transitions from $f(q)$ are  $f(q)\tran{a?}f(q)$ (using rule $r$) and $f(q)\tran{a?}q$ (using rule $r''$). Since $0~\niocos~f(q)$ and $0~\niocos~q$ (because $a?\notin\ins(0)$), this implies that $f(p)\niocos f(q)$. 
\end{example}

\begin{example}
	This example indicates that the use of output actions in negative premises of output-emitting rules would invalidate Theorem~\ref{thm:congruence}. Let $f$ be defined by the following rules:
	\begin{mathpar}
		\inferrule*[]
		{x\tran{a!}y}
		{f(x)\tran{\delta!}f(x)}
		\and
		\inferrule*[right={\normalsize ,}]
		{x\notran{a!}}
		{f(x)\tran{a!}f(x)}
	\end{mathpar}
	
	\noindent where $a!\in O$. Now, let us consider the following processes:
	
	{\centering
		\begin{tikzpicture}[->,>=stealth',shorten >=1pt,auto]
		\matrix [matrix of math nodes, column sep={4cm,between origins},row sep={1cm,between origins}]
		{
			\node (C0) {p}; & \node (D0) {q};\\
		};
		\begin{scope}[every node/.style={font=\small\itshape}]
		\path   
		(C0)   edge[loop left]	   node  {$b!$} (C0)
		(D0)   edge[loop right]    node  {$a!$} (D0)
		(D0)   edge[loop left]     node  {$b!$} (D0);
		\end{scope}
		\end{tikzpicture}
		
	}
	
	\noindent where $b!\in O$. Again, $p\iocos q$ but $f(p)\niocos f(q)$, because $f(p)\tran{a!}f(p)$ but $f(q)\notran{a!}$.
\end{example}

\begin{example}
	This example indicates that the use of input actions in positive premises of output-emitting rules would invalidate Theorem~\ref{thm:congruence}. Let $f$ be defined by the following rules:
	\begin{mathpar}
		\inferrule*[]
		{x\notran{b?}}
		{f(x)\tran{\delta!}f(x)}
		\and
		\inferrule*[right={\normalsize ,}]
		{x\tran{b?}y}
		{f(x)\tran{a!}f(x)}
	\end{mathpar}
	
	\noindent where $a!\in O$ and $b?\in I$. Now, let us consider the following processes:
	
	{\centering
		\begin{tikzpicture}[->,>=stealth',shorten >=1pt,auto]
		\matrix [matrix of math nodes, column sep={4cm,between origins},row sep={1cm,between origins}]
		{
			\node (C0) {p}; & \node (D0) {q};\\
		};
		\begin{scope}[every node/.style={font=\small\itshape}]
		\path   
		(C0)   edge[loop left]	   node  {$\delta!$} (C0)
		(C0)   edge[loop right]    node  {$b?$} (C0)
		(D0)   edge[loop right]    node  {$\delta!$} (D0);
		\end{scope}
		\end{tikzpicture}
		
	}
	
	\noindent Again, $p\iocos q$ but $f(p)\niocos f(q)$, because $f(p)\tran{a!}f(p)$ but $f(q)\notran{a!}$.
\end{example}

\subsection{Applying modal decomposition}\label{ssec:decomposition}

Assume that we want to know whether $f(p_1,\ldots,
p_n)\models\varphi$, with $f$ specified by rules in \iocos-format and
$\varphi\in\Liocos$. One could construct the LTS for $f(p_1,\ldots,
p_n)$ from those for $p_1,\ldots, p_n$ and then use the rules defining
the satisfaction relation $\models$ to check whether $f(p_1,\ldots,
p_n)$ satisfies $\varphi$. However, as is well known, this approach
suffers from the so-called state-explosion problem. Alternatively, one
can apply a compositional approach: to check whether $f(p_1,\ldots,
p_n) \models \varphi$, one first constructs, from $\varphi$ and the
rules defining the operation $f$, a collection of properties
$\varphi_1, \ldots, \varphi_n$ such that $f(p_1,\ldots,
p_n)\models\varphi$ if, and only if, $p_i\models \varphi_i$ for
$i\in\{1, \ldots, n\}$; after that, one checks whether each statement
$p_i\models \varphi_i$ holds. (Even though compositional model
checking suffers from the `formula-explosion problem' in the worst
case and is therefore no panacea in general, variations on that
approach have been applied successfully in the literature---see, for
instance,~\cite{Andersen95,LaroussinieL98}.)

Here we follow this compositional approach by applying the so-called
modal decomposition method of~\cite{BFG04,FokkinkGW06}. In those
papers, ruloids play an important role. For a GSOS languages, a ruloid
\begin{equation}\label{eqn:gsos-ruloid}
\sosrule{\{x_{i} \tran{a_{ij}} y_{ij}  \mid 1 \leq i \leq n, 1 \leq j \leq m_i\} \cup \{x_{i} \notran{b_{ik}}  \mid 1 \leq i 
\leq n, 1 \leq k \leq \ell_i \}}
{D[\arr{x}] \tran{a} C[\vec{x}, \vec{y}]}
\end{equation}
where the $x_i$'s and the $y_{ij}$'s $(1 \leq i \leq n$ and $1 \leq j \leq m_i)$ are all distinct variables, $m_{i}$ 
and $l_{i}$ are natural numbers, $D[\arr{x}]$ is a term with variables including at most the $x_{i}$'s, 
$C[\vec{x}, \vec{y}]$ is a term with variables including at most the $x_{i}$'s and $y_{ij}$'s, and the $a_{ij}$'s, 
$b_{ik}$'s and $a$ are actions from $L$. 

In what follows,
we will limit ourselves to considering only operations in
\iocos-format specified by rules whose target is either a variable or
a term of the form $g(z_1,\ldots, z_n)$, and the metavariables $t$ and
$u$ will range over terms of those forms.  The extension of our
results to arbitrary rules in \iocos-format can be carried out along
the lines in~\cite{BFG04,FokkinkGW06}.

In the light of our simplified setting, \iocos\ rules, i.e., GSOS
rules in \iocos-format, play the role of ruloids
in~\cite{BFG04,FokkinkGW06}; the only other ruloids we need are those
for variables, viz.
\begin{mathpar}
	\inferrule*[right={\normalsize $a\in L$.}]
	{x\tran{a}x'}
	{x\tran{a}x'}
\end{mathpar}
In what follows, we often refer to ruloids as rules.

A crucial result from~\cite{BloomThesis,BIM95} is that a transition
$\sigma(t)\tran{a}p'$ is provable in a GSOS language if, and only if,
there exist a substitution $\sigma'$ and a GSOS ruloid $H/t\tran{a}u$
such that $\sigma'$ satisfies all premises in $H$ (denoted by
$\sigma'\models H$), $\sigma'(t)=\sigma(t)$ and $\sigma'(u)=p'$.

Before defining the modal decomposition for $\Liocos$ in general, let
us describe how to define it for the modal operator $\eexis{a?}$, with
$a?\in I$. First, recall that $p\models\eexis{a?}\varphi$ iff either
$p\notran{a?}$ or there exists some $p'$ such that $p\tran{a?}p'$ with
$p'\models\varphi$. Hence, the modal decomposition must consider both
cases: when is possible to apply a rule that gives rise to an
$a?$-transition, and when it is not. This first case is dealt with
following the definition of the decomposition of formulae in HML given
in~\cite{FokkinkGW06}. For the second case, we first define the
following auxiliary notations and concepts.

\begin{definition}\label{def:decomp_no_tran}
Let $P$ be a GSOS language, $t$ a term, and $a?\in I$ an input
action. We write $R(t,a?)$ for the set of rules for $t$ that emit
$a?$. That is,
\[R(t,a?)= \{r\in P\mid\exists H,u.~ r=H/t\tran{a?}u\}. 
\]  

We write $H_r$ for the set of premises of rule $r\in R(t,a?)$.  Given a premise $\gamma\in H_r$ and a 
variable $x$, we define the formula $\nega(\gamma,x)\in \Liocos$ in the following way:
\begin{itemize}
	\item $\nega(x\notran{b!}, x)= \exis{b!}\true$, with $b!\in O$.
	
	\item $\nega(x\tran{b?}x', x)= \eexis{b?}\false$, with $b?\in I$.
	
	\item $\nega(y\notran{b!}, x)= \nega(y\tran{b?}y', x)= \true$, with $y\neq x$.
\end{itemize}
Finally, we write $\chi(t,a?)$ for the set of all the functions that pick a premise in $H_r$ for each $r \in R(t,a?)$---that is,

\[\chi(t,a?)= \{\eta\mid \eta:R(t,a?)\tran{}\bigcup_{r\in R(t,a?)}{H_r},\,\text{such that $\eta(r)\in H_r,\forall r\in 
R(t,a?)$}\}.
\]
\end{definition}

\begin{remark}\label{rem:notran}
Intuitively, $\nega(\gamma,x)$ gives us a \iocos\ formula that
captures that the premise $\gamma$ is not satisfied. For example, if
$\gamma= x\notran{b!}$, then $\nega(x\notran{b!}, x)= \exis{b!}\true$
and given a closed substitution $\sigma$, if $\sigma(x)\models
\exis{b!}\true$, then $\sigma(x)$ does not satisfy the premise
$\gamma$.
\end{remark}

The following example illustrates the discussion in Remark~\ref{rem:notran}.

\begin{example}\label{ex:decomp_no_tran}
Let us consider $t= f(x,y)$, with rules

\begin{mathpar}
	\inferrule*[left=$r_1$]
	{x\notran{a!} }
	{f(x,y)\tran{a?}x}
	\and
	\inferrule*[right={\normalsize .}, left=$r_2$]
	{x\tran{a?}x'\quad y\tran{b?}y'}
	{f(x,y)\tran{a?}x'}
\end{mathpar}

We want to characterise when $\sigma(f(x,y))\notran{a?}$, for each closed substitution $\sigma$. In order to 
do so, first let us define $\chi(f(x,y),a?)$. Since there are two premises in $H_{r_2}$, $\chi(f(x,y),a?)=\{\eta_1, 
\eta_2\}$, where:
\begin{displaymath}
\begin{array}{rcl}
\eta_1 & = & \{r_1\mapsto x\notran{a!}, r_2\mapsto x\tran{a?}x'\}\, ,\\
\eta_2 & = & \{r_1\mapsto x\notran{a!}, r_2\mapsto y\tran{b?}y'\}\, .
\end{array}
\end{displaymath}

Now, let us define $\psi_{\eta_i}(z)=\bigwedge_{r\in R(t,a?)} \nega(\eta_i(r),z)$, for $z\in \{x,y\}$. That is, 
working up to logical equivalence, 
\begin{displaymath}
\begin{array}{rclcrcl}
\psi_{\eta_1}(x) & = & \exis{a!}\true \wedge \eexis{a?}\false\, , &
\psi_{\eta_1}(y) & = & \true\, ,\\
\psi_{\eta_2}(x) & = & \exis{a!}\true\, ,&
\psi_{\eta_2}(y) & = & \eexis{b?}\false\, .\\
\end{array}
\end{displaymath}

We observe that if $\sigma(x)\models \psi_{\eta_i}(x)$ and $\sigma(y)\models \psi_{\eta_i}(y)$ for some $i\in 
\{1,2\}$, then $\sigma(f(x,y))\notran{a?}$, and that the converse implication also holds. (See 
Proposition~\ref{lem:no_rules} in~\ref{App:decomposition} for the formal result.)
\end{example}

We are now ready to define modal decomposition \`a la Fokkink and van
Glabbeek for the logic $\Liocos$. We will give the definition for an
alternative syntax of the logic $\Liocos$ where, as noted in
Definition~\ref{Def:HML}, we use finitary conjunctions and
disjunctions, and follow the convention that an empty conjunction
stands for $\true$ and an empty disjunction stands for $\false$.

\begin{definition}[Modal decomposition]\label{def:decomposition}
  The decomposition function
  \[
  \cdot^{-1}:\mathbb{T}(\Sigma)\rightarrow 
  (\Liocos\rightarrow\mathcal{P(\mathsf{Var}\rightarrow\Liocos)})
  \]
  is defined in the following way:
\begin{itemize}
 \item $\psi\in t^{-1}(\eexis{a?}\varphi)$ iff $\psi\in t^{-1}_\chi(\eexis{a?}\varphi) \cup 
 t^{-1}_R(\eexis{a?}\varphi)$ where $t^{-1}_\chi(\eexis{a?}\varphi)$ and $t^{-1}_R(\eexis{a?}\varphi)$ are 
 defined as follows: 
 \begin{itemize}
  \item $\psi\in t^{-1}_\chi(\eexis{a?}\varphi)$ iff exists a function $\eta\in\chi(t,a?)$ such that 
  $\psi=\psi_\eta$, where 	for each $x\in \mathsf{Var}$,
	\[ \psi_\eta(x)=\bigwedge_{r\in R(t,a?)} \nega(\eta(r),x).\]
			
  \item $\psi\in t^{-1}_R(\eexis{a?}\varphi)$ iff there are a rule $r=H/t\tran{a?}t'\in R(t, a?)$ and a 
  decomposition mapping $\psi'\in u^{-1}(\varphi)$ such that, for each $x\in \mathsf{Var}$,
	\begin{displaymath}
	  \psi(x)=\left\{
	  \begin{array}{ll}
	  \bigwedge_{\substack{x\tran{b?}y\in H\\ y\in var(u)}}\eexis{b?}\psi'(y)\wedge \psi'(x) & \text{if $x\in 
	  var(u)$}\\
	  \bigwedge_{\substack{x\tran{b?}y\in H\\ y\in var(u)}}\eexis{b?}\psi'(y) & \text{if $x\notin var(u)$}\\
	  \end{array}
	  \right.
	\end{displaymath}
\end{itemize}
		
\item $\psi\in t^{-1}(\exis{a!}\varphi)$ iff there exist some rule ${H}/{t\tran{a!}u}$ and some decomposition mapping 
$\psi'\in u^{-1}(\varphi)$ such that, for each $x\in \mathsf{Var}$,
\[\psi(x)=\left\{
	\begin{array}{l@{\hspace{0.2cm}}l}
	\bigwedge_{\substack{x\tran{b!}y\in H\\ y\in var(u)}} \exis{b!}\psi'(y)\wedge \bigwedge_{x\notran{c?}\in 
	H}\eexis{c?}\false\wedge   
	\psi'(x) & \text{if $x\in var(u)$}\\
	\bigwedge_{\substack{x\tran{b!}y\in H\\ y\in var(u)} }\exis{b!}\psi'(y)\wedge \bigwedge_{x\notran{c?}\in 
	H}\eexis{c?}\false 
	& \text{if 
	$x\notin var(u)$}\\
    \end{array}
 	\right.
\]
		
 \item $\psi\in t^{-1}(\bigvee_{i\in I}\varphi_i)$ iff $\psi\in\bigcup_{i\in I}t^{-1}(\varphi_i)$ (where $I$ is a finite 
 index set).
		
 \item $\psi\in t^{-1}(\bigwedge_{i\in I}\varphi_i)$ iff there are $\psi_i\in t^{-1}(\varphi_i)$ for $i\in I$, such that 
 for all $x\in \mathsf{Var}$, $\psi(x)=\bigwedge_{i\in I}\psi_i(x)$, where $I$ is finite.
\end{itemize}
\end{definition}

\begin{remark}\label{rem:decomposition}
The only positive premises that are considered when defining $\psi(x)$ are those for which $x\tran{b?}y$ 
with $y\in var(u)$. Indeed, by the definition of GSOS rules, since all the variables are distinct, if 
$x\tran{b?}y\in H$, $\psi(y)=\true$ if $y\notin var(u)$. Also, since $\eexis{a?}\true\equiv\true$, the 
last sentence means that the only positive premises $x\tran{b?}y$ that give a non-trivial $\psi(y)$ are those 
such that $y\in var(u)$. 
\end{remark}

\begin{remark}
Notice that $t^{-1}(\false)=\emptyset$ (case $\bigvee_{i\in I}\varphi_i$ for $I=\emptyset$) and $t^{-1}(\true)=\{\psi\}$, 
where $\psi(x)=\true$ for all $x$ (case $\bigwedge_{i\in I}\varphi_i$ for $I=\emptyset$).
\end{remark}

Our reader may wonder why the definition of $\psi\in
t_R^{-1}(\eexis{a?}\varphi)$ is sufficient to ensure that
$\sigma(t)\models \eexis{a?}\varphi$ whenever $\sigma(x)\models
\psi(x)$ for all $x\in var(t)$. Indeed, the definition in question
does not explicitly require that $\sigma(x)$ satisfy the negative
premises in $H_r$. (Observe that this cannot be done because there is
no formula in $\Liocos$ that is satisfied by all the processes that
cannot perform a $b!$-transition, for some output action $b!$.)
However, as the proof of Theorem~\ref{thm:decomposition} will make
clear (see~\ref{App:decomposition}), since the rules are in {\iocos}
format, it isn't necessary to express the negative premises in the
definition of $\psi\in t_R^{-1}(\eexis{a?}\varphi)$. In fact, if
$\sigma$ does not satisfy any of the mappings $\psi'\in
t_\chi^{-1}(\eexis{a?}\varphi)$, then $\sigma$ satisfies all the
premises of at least one rule $r$ in $R(t,a?)$ and some mapping $\psi$
associated to that rule contains all the information that is needed to
determine whether $\sigma(t)\models \eexis{a?}\varphi$.

All this is formalised in Theorem~\ref{thm:decomposition}, but first we will show some examples of the use of 
the modal decomposition.

\begin{example}
Let us consider the following rules in the \iocos-format (omitting the rules for the quiescence action):

\begin{mathpar}
	\inferrule*[left=$r_1$]
	{ }
	{g(y)\tran{b!}g(y)}
	\and
	\inferrule*[right={\normalsize ,}, left=$r_2$]
	{x\tran{a?}y\quad x\notran{b!}}
	{f(x)\tran{a?}g(y)}
\end{mathpar}
where $a?\in I$ and $b!\in O$. Let us calculate $f(x)^{-1}(\eexis{a?}\exis{b!}\true)$. First, we have that 
$R(f(x),a?)=\{r_2\}$ and, since $H_{r_2}$ has two premises, there are two functions $\chi(t,a?)$, namely $\eta_1(r_2) = 
x\tran{a?}y$ and $\eta_2 (r_2)= x\notran{b!}$. Second, trivially, $g(y)^{-1}(\exis{b!}\true)=\{\phi\}$, where 
$\phi(x)=\true$ for all $x$. 

Now, let us write both $f(x)^{-1}_\chi(\eexis{a?}\exis{b!}\true)$ and $f(x)^{-1}_R(\eexis{a?}\exis{b!}\true)$.
\begin{itemize}
	\item $\psi_{\eta_1}, \psi_{\eta_2}\in f(x)^{-1}_\chi(\eexis{a?}\exis{b!}\true)$, where $\psi_{\eta_1}(x)= 
	\eexis{a?}\false$, $\psi_{\eta_2}(x)= \exis{b!}\true$ and $\psi_{\eta_1}(z)=\psi_{\eta_2}(z)=\true$ for all 
	$z\neq x$.
	
	\item $\psi\in f(x)^{-1}_R(\eexis{a?}\exis{b!}\true)$ iff $\psi(x)= \eexis{a?}\true\equiv \true$ and 
	$\psi(z)=\true$ for all $z\neq x$.
\end{itemize}

This means that any process $p$ is such that $f(p)\models \eexis{a?}\exis{b!}\true$. Indeed, let us assume that $p\tran{a?}$ and $p\notran{b!}$, in this case $f(p)\tran{a?}\sigma(g(y))$ and we are done. On the other hand, if  either $p\notran{a?}$ or $p\tran{b!}$, then $f(p)\notran{a?}$ and therefore $f(p)$ trivially satisfies the formula.
\end{example}

\begin{theorem}
\label{thm:decomposition}
Let $P$ be a GSOS language in \iocos-format. For each term $t$, formula $\varphi$ and closed substitution 
$\sigma$, we have $\sigma(t)\models\varphi$ iff there exists $\psi\in t^{-1}(\varphi)$ such that  $\sigma(x)\models\psi(x)$, for all $x\in 
var(t)$.
\begin{proof}
The proof of this result can be found in \ref{App:decomposition}.
\end{proof}
\end{theorem}

The next counter-example illustrates why the \iocos-format is needed in Theorem~\ref{thm:decomposition}.

\begin{example}
Let $f$ be defined by the following rules, which are not in \iocos-format and are taken from 
Example~\ref{ex:counter-example} on page~\pageref{ex:counter-example}:
\begin{mathpar}
	\inferrule*[left={\normalsize $r_1$}]
	{x\tran{a?}y}
	{f(x)\tran{a?}0}
	\and
	\inferrule*[right={\normalsize .}, left={\normalsize $r_2$}]
	{ }
	{f(x)\tran{a?}f(x)}
\end{mathpar}

Applying the modal decomposition method to obtain $f(x)^{-1}(\eexis{a?}\eexis{a?}\false)$, it is not hard to see that $\psi_{r_1}\in f(x)^{-1}_R(\eexis{a?}\eexis{a?}\false)$, where $\psi_{r_1}(x)=\true$.
The process $s$, with $s\tran{\delta!}s$, satisfies $\true=
\psi_{r_1}(x)$. On the other hand,
$f(s)\not\models\eexis{a?}\eexis{a?}\false$, which means that the
decomposition method is incorrect for this operation.

The rules for $f$ do not meet clause~\ref{postrig} in
Definition~\ref{Def:ruleformat} because $a?$ belongs to the positive
trigger for $x$ in $r_1$ but not to that for $x$ in $r_2$. This clause
is essential in the proof of the ``if implication'' in
Theorem~\ref{thm:decomposition}.
\end{example}


\subsection{A rule format for coherent quiescent behaviour}\label{ssec:rule-quiescent}

Operators for constructing LTSs with inputs and outputs should ensure `coherent quiescent behaviour' in the 
sense of Definition~\ref{dfn:lts}. This means that each operator $f$, when applied to a vector of states 
$\arr{p}$ in an LTS, should satisfy the following property:
\begin{equation}\label{eq:CBQ}
f(\arr{p})\tran{\delta!}p'\text{ iff } p'=f(\arr{p})\text{ and, for each  } a!\in O\setminus\{\delta!\},~ f(\arr{p})\notran{a!} . 
\end{equation}
In what follows, we will isolate sufficient conditions on the GSOS rules defining $f$ that guarantee the 
above-mentioned property.

\begin{definition}\label{def:contradicts}
We say that the following sets of formulae contradict each other:
\begin{itemize}
\item $\{x\tran{a}y\}$ and $\{x\notran{a}\}$ for $a\in L$, 
	
\item $\{x\tran{b!}y\}$ and $\{x\tran{\delta!}z\}$ for $b!\in O\setminus \{\delta!\}$, and 

\item $H$ and $H'$ when $H$ and $H'$ are non-empty and $H \cup H' = \{x\notran{b!}\mid b!\in O\}$.

\end{itemize}
Formulae $x\tran{a}y$ and $x\notran{a}$ are said to negate each other. 

We say that two sets of formulae $H_1$ and $H_2$ are {\em contradictory} if there are $H_1' \subseteq H_1$ and $H_2' 
\subseteq H_2$ such that $H_1'$ and $H_2'$ contradict each other.
\end{definition}

Intuitively, two sets of contradictory formulae cannot be both satisfied by states in an LTS. For example, in the light of the
requirement on quiescent behaviour in Definition~\ref{dfn:lts}, there is no state $p$ in an LTS such that $p\notran{b!}$ for 
each $b!\in O$. This observation motivates the third requirement in Definition~\ref{def:contradicts}.

\begin{definition}\label{def:quiescent_consistent}
We say that an operation $f$ is \emph{quiescent consistent} if the set of rules for $f$ satisfies the following two 
constraints:
\begin{description}
\item[{[$\delta_1$]}] If $H/f(\arr{x})\tran{\delta!}t$ is a rule for $f$ then 
\begin{enumerate}
\item for each $f$-defining rule $H'/f(\arr{x})\tran{b!}t'$ with $b!\in O\setminus \{\delta!\}$, the sets $H$ and 
$H'$ are contradictory, and 

\item $t=f(\arr{y})$ for some vector of variables $\arr{y}$ such that, for each index $i$, either $y_i=x_i$ or 
$x_i \tran{\delta!} y_i\in H$.
\end{enumerate}

\item[{[$\delta_2$]}] Let $\{r_1,\ldots r_n\}$ be the set of output-emitting rules for $f$ not having 
$\delta!$ as label of their conclusions. 

  Then the set of rules for $f$ contains all rules of the form
  \[
  \sosrule{\{l_1,\ldots , l_n\}}{f(\arr{x})\tran{\delta!}f(\arr{x})} ,
  \]
where $l_i$ negates some premise of $r_i$ and no two sets of formulae included in $\{l_1,\ldots l_n\}$ 
contradict each other.
\end{description}

A GSOS language is \emph{quiescent consistent} if so is each operation in it.
\end{definition}

\begin{theorem}\label{thm:quiescent}
If $f$ is quiescent consistent then Property~(\ref{eq:CBQ}) holds for $f$. 
\begin{proof}
The proof of this result may be found in \ref{App:thmquiescent}.  
\end{proof}
\end{theorem}

We end the section showing a negative result for the (binary) merge
operator of Example~\ref{Ex:merge}. In that example, we mentioned
that, as noted in~\cite{BenesDHKN15}, the merge of two systems need
not satisfy Property~(\ref{eq:CBQ}). Here we prove that there is no
extension of the $\delta!$-emitting rules for the merge operator with
rules in \iocos-format that leads to an operation affording
Property~(\ref{eq:CBQ}). The interested reader can find examples of
operators that are both quiescent and \iocos\ conforming in the
paper~\cite{GLM18}.

\begin{proposition}\label{prop:merge}
There is no extension of the rules for the (binary) merge operator with a collection of $\delta!$-emitting rules 
in \iocos-format that guarantees consistent quiescent behaviour.
\end{proposition}

\begin{proof}
Our goal is to add $\delta!$-emitting rules in \iocos-format to those for the merge operator so that, for all $p, 
q$ 
\begin{equation}\label{eq:merge}
{p \wedge q \tran{\delta!} r} \Leftrightarrow {r=p\wedge q \text{ and } p\wedge q \notran{a!} \forall a!\in 
O\setminus \{\delta!\}\;. }
\end{equation}

We will show that this is impossible. Let us recall the output emitting rules for the binary merge operator: 	
\begin{mathpar}
	\inferrule*[right={\normalsize $a!\in O$.}, left={$r_{a!}$}]
	{x\tran{a!}x'\quad y\tran{a!}y'}
	{x\wedge y\tran{a!}x'\wedge y'}
\end{mathpar}	
In particular, $r_{\delta}$ ensures that $p \wedge q \tran{\delta!}p \wedge q$ for all quiescent $p$ and $q$.

Now, assume that there is a set of $\delta!$-emitting rules $R$ in \iocos-format such that $R\cup \{r_{a!}\mid 
a!\in O\}$ ensures Property~(\ref{eq:merge}). Let us consider $a!\tran{a!}0$ and $b!\tran{b!}0$, with $a\neq 
b$. Since $a!\wedge b!\notran{c!}$ for each $c!\in O\setminus\{\delta!\}$ there are some rule $r= H/x\wedge 
y\tran{\delta!}t$ and closed substitution $\sigma$ such that:
\begin{inparaenum}[(i)]
	\item $\sigma(x)=a!$, $\sigma(y)=b!$;
	\item $\sigma\models H$; and
	\item $\sigma(t)=a!\wedge b!$.
\end{inparaenum}

Since $r$ is in \iocos-format and $\sigma\models H$, its positive
premises can only have the form $x\tran{a!}x'$ for some $x'$ and
$y\tran{b!}y'$ for some $y'$. The negative premises of $H$, if any,
have the form $x\notran{c?}$ or $y\notran{c?}$ for some $c?\in
I$. Moreover, as $a!\tran{a!}0$, $b!\tran{b!}0$ and
$\sigma(t)=a!\wedge b!$, none of the variables $x'$ and $y'$ can occur
in $t$. So $t$ can only have the form $x\wedge y$, $a!\wedge y$,
$x\wedge b!$ or $a!\wedge b!$. We claim that it must be $x\wedge y$.
Indeed, assume, by way of example, that $t=a!\wedge y$. Consider now
the substitution $\sigma'(x)=a!+b!$ and $\sigma'(y)=b!$, where
$a!+b!\tran{a!}0$ and $a!+b!\tran{b!}0$. It is straightforward to
check that $\sigma'\models H$. Hence, $r$ yields $(a!+b!)\wedge
b!\tran{\delta!}a!\wedge b!$, which contradicts
Property~(\ref{eq:merge}). Similar examples can be constructed if
$t=x\wedge b!$ and $t=a!\wedge b!$.

Thus, we have $r= H/x\wedge y\tran{\delta!}x\wedge y$. Consider now the substitution $\rho$ such that 
$\rho(x)=\rho(y)=a!+b!$ and $\rho(z)=0$ for $z\notin\{x,y\}$. This substitution 
satisfies all the positive premises in $H$ and, since all the negative
premises in $H$ are of the form $y\notran{c?}$ for some $c?\in I$,
$\rho$ also satisfies those. Hence $\rho\models H$, and $\rho$
together with $r$ proves $(a!+b!)\wedge
(a!+b!)\tran{\delta!}(a!+b!)\wedge (a!+b!)$.

On the other hand, using $r_a$, we also have that $(a!+b!)\wedge (a!+b!)\tran{a!}0\wedge 0$, which 
contradicts Property~(\ref{eq:merge}).

It follows that no extension of the rules for the (binary) merge
operator with $\delta!$-emitting rules in \iocos-format satisfies
Property~(\ref{eq:merge}).
\end{proof}

\section{Conclusion}\label{sec:conclusion}

In this paper, we have further developed the theory of
{\iocos}~\cite{GLM13,GLM14,GLM15} by studying logical
characterisations of this relation and its compositionality. We have
also compared the proposed logical characterisation of {\iocos} with
an existing logical characterisation for {\ioco} proposed by Beohar
and Mousavi. The article also offers a precongruence rule format for
{\iocos} and a rule format ensuring that operations take quiescence
properly into account. Both rule formats are based on the GSOS format
by Bloom, Istrail and Meyer.

We have provided a connection between the precongruence rule format
for {\iocos} and logical characterisations for that relation by
establishing a modal decomposition result, which yields a
compositional model-checking procedure for the studied logic with
respect to systems described using operations in {\iocos}-format. We
have also studied fixed-point extensions of the logics introduced in
this paper and offered a characteristic-formula construction for
processes modulo {\iocos}.

In future research, it would be interesting to study whether the
precongruence rule format we provide in this paper can be made more general and
whether it can be derived from a modal decomposition result in the
style of Fokkink and van Glabbeek. Our modal decomposition result
relies heavily on properties of the {\iocos} format; we believe
therefore that generalizing the precongruence rule format might be
difficult.

In this paper, we have considered the theory of {\iocos} as presented
in \cite{GLM13,GLM14,GLM15,GLM18}. That theory does not consider the
internal action $\tau$, which is the source of many of the
complications in the classic {\ioco} testing theory. Extending the
theory of {\iocos} to a setting with the internal action $\tau$ is an
interesting avenue for future research. As a first step, one could
develop a version of {\iocos} with the action $\tau\not\in L$, but
without abstracting from internal steps in system executions. This can
be done by requiring that condition~\ref{dfn:iocos:3} in the
definition of an \iocos-relation (Definition~\ref{dfn:iocos}) hold
also with respect to $\tau$-transitions---that is, that internal
transitions performed by implementations should be matched by
specifications as in the classic simulation preorder~\cite{Mil71}. In
this framework, we could replay the theory developed in this paper. By
way of example, the rule format given in
Definition~\ref{Def:ruleformat} can be extended to the resulting
version of {\iocos} by requiring that $\tau$-emitting rules have only
input actions as labels of negative premises and output actions as
labels of positive premises. The resulting rule format would be
powerful enough to express the hiding operation considered
in~\cite{BijlRT03}. Of course, the ultimate goal of extending {\iocos}
with internal actions would be to develop a theory that abstracts from
$\tau$ like {\ioco} and that has some of the pleasing properties of
that classic notion of conformance, including a test generation
procedure. We consider the study of such a theory and of its potential
applications a worthy research goal for the future.

%

\newpage 
\appendix 

\section{Proof of Theorem~\ref{thm:congruence}}\label{App:thmcongruence}

\begin{remark}
All the substitutions we consider from now on map variables either to
variable-free terms or to states in some LTS.
\end{remark}

Let $\mathcal{S}$ be the least binary relation that includes
$\iocos$ and such that: if $p_i \mathbin{\mathcal{S}} q_i$, for each $1 \leq i \leq n$, and $f$ is an $n$-ary operation in \iocos-format then 
$f(p_1,\ldots, p_n) \mathbin{\mathcal{S}} f(q_1,\ldots, q_n)$.

The following property of $\mathcal{S}$ will be useful in what follows.
\begin{lemma}\label{lem:subst}
Let $t$ be a term. Assume that $\sigma$ and $\rho$ are two
substitutions such that $\sigma(x) \mathbin{\mathcal{S}} \rho(x)$ for
each variable $x$ occurring in $t$. Then
$\sigma(t) \mathbin{\mathcal{S}} \rho(t)$.
\end{lemma}

\begin{proof}[\textbf{Proof of Theorem~\ref{thm:congruence}}]
We show that $\mathcal{S}$ is an \iocos\ simulation by induction on the definition of $\mathcal{S}$. This is clear if $p
\mathbin{\mathcal{S}} q$ because $p \iocos q$. Assume
therefore that
\[
p = f(p_1,\ldots, p_n) \mathbin{\mathcal{S}}  f(q_1,\ldots, q_n) = q 
\]
because $p_i \mathbin{\mathcal{S}} q_i$ for each $1 \leq i \leq n$. We
proceed to prove that each of the defining conditions for \iocos relations holds for $p$ and $q$.  In doing so, we shall use, as our
inductive hypothesis, the fact that those conditions hold for $p_i
\mathbin{\mathcal{S}} q_i$ ($1 \leq i \leq n$). 
In particular, we have that $\ins(q_i)\subseteq \ins(p_i)$ ($1 \leq i \leq n$).

\begin{itemize}
\item We prove, first of all, that the set of initial input actions of $q$ is included in the set of initial input actions of $p$. To this end, assume that $a?$ is an input action and $q = f(q_1,\ldots, q_n) \tran{a?}{}$. We shall prove that 
$p = f(p_1,\ldots,p_n)\tran{a?}{}$ also holds.

Since $q \tran{a?}{}$, there are a rule $r$ of the form (\ref{eqn:gsos-rule}) on page~\pageref{eqn:gsos-rule} and substitution $\sigma$ such that
\begin{itemize}
\item $\sigma(x_i) = q_i$ for each $1 \leq i
  \leq n$, 
\item $q_i \tran{a_{ij}} \sigma(y_{ij})$ for each $1 \leq i
  \leq n$ and $1 \leq j \leq m_i$ and 
\item $q_i \notran{b_{ik}}{}$ for each  $1 \leq i \leq n$ and  $1 \leq k \leq \ell_i$. 
\end{itemize}
As $a?$ is an input action, by requirement~\ref{IC1} in
Definition~\ref{Def:ruleformat}, we have that each $a_{ij}$ is an
input action and each $b_{ik}$ is an output action. Since $p_i
\mathbin{\mathcal{S}} q_i$ ($1 \leq i \leq n$), by the inductive
hypothesis, we therefore have that:
\begin{itemize}
\item for each $1 \leq i
  \leq n$ and $1 \leq j \leq m_i$, as
  $a_{ij}\in\ins(q_i)\subseteq \ins(p_i)$, there is some state
  $p_{ij}$ such that $p_i \tran{a_{ij}}  p_{ij}$ and
\item $p_i \notran{b_{ik}}{}$, for each  $1 \leq i \leq n$ and  $1 \leq k \leq \ell_i$. 
\end{itemize}
Therefore rule $r$ can be used to infer that $p \tran{a?}{}$, and we are done. 

\item Assume now that $p = f(p_1,\ldots, p_n) \tran{a?} p'$ and $a?$ is an
  initial input action of $q = f(q_1,\ldots, q_n)$. As $a?$ is an
  initial input action of $q = f(q_1,\ldots, q_n)$, there are a rule
  $r = \frac{H}{f(x_1,\ldots,x_n) \tran{a?} t}$ of the form
  (\ref{eqn:gsos-rule}) on page~\pageref{eqn:gsos-rule} and a
  substitution $\sigma$ such that
\begin{itemize}
\item $\sigma(x_i) = q_i$ for each $1 \leq i
  \leq n$, and 
\item $\sigma$ satisfies all the premises in $H$. 
\end{itemize} 
Moreover, there are a
  rule $r' =
    \frac{H'}{f(x_1,...,x_n) \tran{a?} t'}$ of the form (\ref{eqn:gsos-rule}) on
  page~\pageref{eqn:gsos-rule} and a substitution $\sigma'$ such that
\begin{itemize}
\item $\sigma'(x_i) = p_i$ for each $1 \leq i
  \leq n$, 
\item $\sigma'$ satisfies all the premises in $H'$, and 
\item $\sigma'(t') = p'$. 
\end{itemize}
Our goal is to prove that $q = f(q_1,\ldots, q_n) \tran{a?} q'$ for
some $q'$ such that $p' \mathbin{\mathcal{S}} q'$. 
To this end, observe, first of all, that requirement~\ref{IC2} in Definition~\ref{Def:ruleformat} tells us that there is an $f$-defining  rule $r'' =
      \frac{H''}{f(x_1,...,x_n) \tran{a?} t'}$ such that
        \begin{itemize}
        \item for each $1\leq i\leq n$, the positive
          trigger for variable $x_i$ in $r''$ is included in the
          positive trigger for variable $x_i$ in $r$; 
        \item for each $1\leq i\leq n$, the negative
          trigger for variable $x_i$ in $r''$ is included in the
          negative trigger for variable $x_i$ in $r$;   
        \item if $x_i \tran{b?} z$ is contained in
          $H''$ and $z$ occurs in $t'$, then $x_i \tran{b?} z$ is also
          contained in $H'$.
        \end{itemize}
To complete the proof for this case, we will now show how to use the
rule $r''$ to prove the required transition $q = f(q_1,\ldots,
q_n) \tran{a?} q'$. To this end, we will construct a substitution
$\rho$ with the following properties:
\begin{enumerate}
\item $\rho(x_i) = q_i$ for each $1\leq i\leq n$, 
\item $\sigma'(z) \mathbin{\mathcal{S}} \rho(z)$ for each variable $z$ occurring in $t'$ and 
\item $\rho$ satisfies all the premises in $H''$. 
\end{enumerate}
The first condition above gives us that $\rho(f(x_1,\ldots,x_n)) =
q$. The second yields that $p'
= \sigma'(t') \mathbin{\mathcal{S}} \rho(t')$ by
Lemma~\ref{lem:subst}. From the third, we obtain that
\[
q = f(q_1,\ldots, q_n) \tran{a?} \rho(t') . 
\]
Therefore, the substitution $\rho$ and the rule $r''$ prove the
existence of the required transition $q = f(q_1,\ldots,
q_n) \tran{a?} \rho(t') =q'$ with $p' \mathbin{\mathcal{S}} q'$.

To meet the first condition above, we start by setting $\rho(x_i) =
q_i$ for each $1\leq i\leq n$.  Note, next, that, since the negative
trigger for variable $x_i$ in $r''$ is included in the negative
trigger for variable $x_i$ in $r$ ($1\leq i\leq n$), and $\sigma$
satisfies $H$, any substitution $\rho$ such that $\rho(x_i) = q_i$
($1\leq i\leq n$) meets all the negative premises in $H''$. Our aim
now is to extend the definition of $\rho$ to all the other variables
occurring in rule $r''$ in such a way that the other two
above-mentioned requirements are met. To this end, consider a positive
premise $x_i \tran{b?}  z\in H''$. We know that the positive trigger
for variable $x_i$ in $r''$ is included in the positive trigger for
variable $x_i$ in $r$. Therefore $H$ contains a positive premise
$x_i \tran{b?}  w$ for some variable $w$. As $\sigma$ satisfies $H$,
we have that $\sigma(x_i)=q_i \tran{b?}  \sigma(w)$ and therefore
$b?\in\ins(q_i)$. If $z$ does not occur in $t'$, setting $\rho(z)
= \sigma(w)$ will satisfy the premise $x_i \tran{b?}  z\in
H''$. Assume now that $z$ does occur in $t'$. In this case, by
requirement~\ref{posprem} in the definition of the rule format
(Definition~\ref{Def:ruleformat}), we know that $x_i \tran{b?}  z$ is
also contained in $H'$. Since $\sigma'$ satisfies $H'$, we have that
\[
\sigma'(x_i) = p_i \tran{b?} \sigma'(z) . 
\]
As $p_i \mathbin{\mathcal{S}} q_i$, $b?\in\ins(q_i)$ and $\sigma'(x_i) = p_i \tran{b?} \sigma'(z)$, the inductive hypothesis yields that 
\[
\rho(x_i) = q_i  \tran{b?} q_i' \text{ for some $q'_i$ such that } 
\sigma'(z) \mathbin{\mathcal{S}} q_i' . 
\]
We can therefore set $\rho(z) = q'_i$ in order to keep meeting the last
two requirements on $\rho$.

Continuing in this fashion until we have exhausted all the positive premises in $H''$ completes the proof for this case. 

\item Assume that $p = f(p_1,\ldots, p_n) \tran{a!} p'$ for some
  output action $a!$ and state $p'$. Then there are a rule $r$ of the
  form (\ref{eqn:gsos-rule}) on page~\pageref{eqn:gsos-rule} and a
  substitution $\sigma$ such that
\begin{itemize}
\item $\sigma(x_i) = p_i$ for each $1 \leq i
  \leq n$, 
\item $p_i \tran{a_{ij}} \sigma(y_{ij})$ for each $1 \leq i
  \leq n$ and $1 \leq j \leq m_i$, 
\item $p_i \notran{b_{ik}}{}$ for each  $1 \leq i \leq n$ and  $1 \leq k \leq \ell_i$, and 
\item $\sigma(C[\vec{x}, \vec{y}]) = p'$. 
\end{itemize}
Our goal is to use rule $r$ to prove that $q = f(q_1,\ldots,
q_n) \tran{a!} q'$ for some $q'$ such that $p' \mathbin{\mathcal{S}}
q'$.

Since $a!$ is an output action, condition~\ref{IC3} in
Definition~\ref{Def:ruleformat} tell us that each $a_{ij}$ is an output
action, and each $b_{ik}$ is an input action. As each $a_{ij}$ is an
output action, $p_i \tran{a_{ij}!} \sigma(y_{ij})$ and $p_i
\mathbin{\mathcal{S}} q_i$ ($1 \leq i \leq n$), the inductive
hypothesis yields that for each $i$ and $j$ there is some $q_{ij}$
such that $q_i \tran{a_{ij}!} q_{ij}$ and $\sigma(y_{ij})
\mathbin{\mathcal{S}} q_{ij}$.  Note, moreover, that the $q_i$'s
satisfy the negative premises of rule $r$. Indeed,
$p_i \notran{b_{ik}}$ ($1 \leq i \leq n$ and $1 \leq k \leq \ell_i$)
and, as $p_i
\mathbin{\mathcal{S}} q_i$ ($1 \leq i \leq n$), the inductive
hypothesis gives us that the set of input actions of each $q_i$ is
included in that of $p_i$. Therefore, we have that rule $r$
instantiated with a substitution $\rho$ such that
\begin{eqnarray*}
\rho(x_i) & = & q_i\quad (1\leq i \leq n),\, \text{and}\\
\rho(y_{ij}) & = & q_{ij}\quad (1\leq i \leq n, 1\leq j \leq m_i) 
\end{eqnarray*}
yields the transition 
\[
q = f(q_1,\ldots, q_n)\tran{a!} \rho(C[\vec{x}, \vec{y}]) . 
\]
By construction, $\sigma(z) \mathbin{\mathcal{S}} \rho(z)$ for each
variable $z$ occurring in $C[\vec{x}, \vec{y}]$. 
Therefore
Lemma~\ref{lem:subst} now yields that
\[
p' = \sigma(C[\vec{x}, \vec{y}]) \mathbin{\mathcal{S}} 
\rho(C[\vec{x}, \vec{y}]), 
\]
and we are done. \qedhere
\end{itemize}
\end{proof}

\section{Proof of Theorem~\ref{thm:decomposition}}\label{App:decomposition}

Before proving Theorem~\ref{thm:decomposition}, we will show two
auxiliary lemmas that will be use in its proof.

\begin{lemma}\label{lem:no_rules}
Let $P$ be a GSOS language in \iocos-format. For every term $t$, input
$a?\in I$ and closed substitution $\sigma$, we have that
$\sigma(t)\notran{a?}$ if, and only if, there exists some
$\eta\in \chi(t,a?)$ such that $\sigma(x)\models \psi_\eta(x)$ for
each $x\in var(t)$.
\end{lemma}

\begin{proof}
The result follows straightforwardly. First, recall that, by
definition, $\psi_\eta(x)=\bigwedge_{r\in R(t,a?)}
\nega(\eta(r),x)$. Second, note that given a premise $\eta(r)\in H_r$, the formula $\nega(\eta(r),x)$ is  
such that $\sigma(x)\models \nega(\eta(r),x)$ iff $\sigma(x)$ does not
satisfy premise $\eta(r)$. So, if there is some $\eta\in \chi(t,a?)$
such that $\sigma(x)\models \psi_\eta(x)$, for each $x\in var(t)$,
then no rule in $R(t,a?)$ can be used to derive a transition from
$\sigma(t)$. Conversely, if $\sigma(t)\notran{a?}$, then, for each
$r\in R(t,a?)$, there is some premise in $H_r$ that is not satisfied
by $\sigma(t)$. Let $\eta$ be the choice function that picks those
premises. Then, by the discussion above,
$\sigma(x)\models \psi_\eta(x)$ for each $x\in var(t)$.
\end{proof}

\begin{lemma}\label{lem:same_decomposition}
Let $P$ be a GSOS language in \iocos-format. Let $r= H/t\tran{a?}u$ and $r'= H'/t\tran{a?}u$ be such that if 
$x\tran{b?}y\in H'$ with $y\in var(u)$, then $x\tran{b?}y\in H$. Consider $\psi,\psi'\in t^{-1}_R(\eexis{a?}\varphi)$ 
and $\phi\in u^{-1}(\varphi)$ such that $\psi$ is obtained using ruloid $r$ and $\phi$, and $\psi'$ is obtained 
using ruloid $r'$ and $\phi$. Then $\psi(x)$ implies $\psi'(x)$ for all $x\in var(t)$. 
\end{lemma}

\begin{proof}
We show that each conjunct of $\psi'(x)$ is implied by $\psi(x)$. Without loss of generality, let us assume 
$x\in var(u)$. By construction,
\begin{itemize}
\item $\psi(x)= \bigwedge_{\substack{x\tran{b?}y\in H\\ y\in var(u)}}\eexis{b?}\phi(y)\wedge \phi(x)$.

\item $\psi'(x)= \bigwedge_{\substack{x\tran{c?}z\in H'\\ z\in var(u)}}\eexis{c?}\phi(z)\wedge \phi(x)$.
\end{itemize}
By hypothesis, if $x\tran{c?}z\in H'$ with $z\in var(u)$, then $x\tran{c?}z\in H$. This implies that each 
conjunct $\eexis{c?}\phi(z)$ of $\psi'(x)$ is also a conjunct of $\psi(x)$, and the claim follows.
\end{proof}

\begin{proof}[\textbf{Proof of Theorem~\ref{thm:decomposition}}]
The proof is by structural induction on the formula $\phi$. We only show the case of $\phi=\eexis{a?}\varphi$ 
since the others are the same as in \cite{FokkinkGW06}.
	
\emph{Only if implication}: Let us suppose that $\sigma(t)\models\eexis{a?}\varphi$. Then, by definition of 
the satisfaction relation $\models$, either $\sigma(t)\notran{a?}$ or
there exists some $p$ such that $\sigma(t)\tran{a?}p$ and
$p\models\varphi$.

Let us first consider the case when $\sigma(t)\notran{a?}$. This means
that it is not possible to apply any of the rules with conclusion
$t\tran{a?}u$ for some $u$.
By Lemma~\ref{lem:no_rules}, there exist $\eta\in \chi(t,a?)$
and $\psi_{\eta}$ such that $\sigma(x)\models \psi_\eta(x)$ for each
$x\in var(t)$. That is, there exists some $\psi_\eta\in
t^{-1}_\chi(\eexis{a?}\varphi)$ such that
$\sigma(x)\models\psi_\eta(x)$, for all $x\in var(t)$, and we are done.
	
On the other hand, if there exists some $p$ such that $\sigma(t)\tran{a?}p$ and $p\models\varphi$, there 
are a rule $r=H/t\tran{a?}u$ and a closed substitution $\sigma'$ such that $\sigma'\models H$,  
$\sigma'(t)=\sigma(t)$ and $\sigma'(u)=p$. Since $\sigma'(u)\models\varphi$ by the induction hypothesis 
there exists some $\phi\in u^{-1}(\varphi)$ such that $\sigma'(z)\models\phi(z)$ for all $z\in var(u)$. Now, let 
us consider $\psi\in t^{-1}_R(\eexis{a?}\varphi)$ constructed as in Definition~\ref{def:decomposition} from the 
rule $r=H/t\tran{a?}u$ and $\phi\in u^{-1}(\varphi)$.
	
In this case we have that $\sigma(x)=\sigma'(x)\tran{b?}\sigma'(y)$
and $\sigma(x)=\sigma'(x)\notran{c!}$, for all $x\tran{b?}y\in H$ and
$x\notran{c!}\in H$. We show that $\sigma(x)=\sigma'(x)\models\psi(x)$ for all
$x\in var(t)$. We distinguish two cases depending on whether $x\in
var(u)$ or not.
	
\begin{itemize}
	\item If $x\in var(u)$, by definition $\psi(x)=\bigwedge_{\substack{x\tran{b?}y\in H, y\in 
	var(u)}}\eexis{b?}\phi(y)\wedge \phi(x)$, and by the inductive hypothesis $\sigma'(x)\models\phi(x)$ and 
	$\sigma'(y)\models\phi(y)$, for each $y\in var(u)$. Since $\sigma'(x)\tran{b?}\sigma'(y)$, because $\sigma'$ satisfies $H$, we 
	obtain the following: $\sigma'(x)\models\bigwedge_{\substack{x\tran{b?}y\in H, y\in var(u)}} 
	\eexis{b?}\phi(y)$. Hence, $\sigma(x)=\sigma'(x)\models\psi(x)$.
		
	\item If $x\notin var(u)$, by definition that $\psi(x)=\bigwedge_{\substack{x\tran{b?}y\in H, y\in var(u)}} 
	\eexis{b}\phi(y)$, and using the same reasoning as before, we conclude $\sigma(x)=\sigma'(x)\models\psi(x)$.
\end{itemize}
Finally, for all $x\in var(t)$, $\sigma(x)=\sigma'(x)\models\psi(x)$.\\
	
	
\emph{If implication}: Let us suppose that $t = f(x_1,\ldots,x_n)$ and that there is some $\psi\in t^{-1}(\eexis{a?}\varphi)$ such that 
$\sigma(x)\models\psi(x)$ for all $x\in var(t)$. We have to show that
$\sigma(t)\models\eexis{a?}\varphi$.

Since the claim holds if $\sigma(t)\notran{a?}$, we need only consider the case that $\sigma(t)\tran{a?}$. 
This means that there are a rule $r= H/t\tran{a?}u$ and a substitution $\sigma_r$ such that 
\begin{inparaenum}[(a)]
\item\label{item:A} $\sigma_r\models H$, and
\item\label{item:B} $\sigma(t) =\sigma_r(t)$, that is, $\sigma(x)=\sigma_r(x)$ for all $x\in var(t)$.
\end{inparaenum}
Since $\sigma(x)\models\psi(x)$ for all $x\in var(t)$, by Lemma~\ref{lem:no_rules}, item~(\ref{item:A}) above yields 
that $\psi\in t^{-1}_R(\eexis{a?}\varphi)$.

Suppose that $\psi$ is constructed using a rule $r'=H'/t\tran{a?}u'$ and $\phi\in u'(\varphi)$. Since we 
assume that the GSOS language is in \iocos-format, condition $(\ref{IC2})$ in Definition~\ref{Def:ruleformat}, 
tells us that there is a rule $r''= H''/t\tran{a?}u$ such that:
\begin{enumerate}[$(i)$]
\item\label{item:pos} for all $i\in\{1,\ldots, n\}$, the positive trigger for variable $x_i$ in $r''$ is included in 
the	positive trigger for variable $x_i$ in $r$;

\item\label{item:neg} for all $i\in\{1,\ldots, n\}$, the negative trigger for variable $x_i$ in $r''$ is included in 
the	negative trigger for variable $x_i$ in $r$; and

\item\label{item:u} if $x_i \tran{b?} z\in H''$ and $z\in var(u)$, then $x_i \tran{b?} z$ is also contained in 
$H'$. 		
\end{enumerate}

Let $\psi'\in t^{-1}_R(\eexis{a?}\varphi)$ be constructed from $r''$ and $\phi\in u'^{-1}(\varphi)$. By 
Lemma~\ref{lem:same_decomposition}, $\psi(x)\Rightarrow\psi'(x)$ for each $x\in var(t)$. Therefore, 
$\sigma(x)=\sigma_r(x)\models\psi'(x)$, for all $x\in var(t)$.

Our goal is to use rule $r''$ to show that $\sigma(t)\tran{a?}p$ for some $p$ such that $p\models\varphi$. To 
this end, we will construct a substitution $\sigma''$ with the following properties:
\begin{inparaenum}[(1)]
	\item\label{q1} $\sigma''(t)=\sigma(t)$,
	\item\label{q2} $\sigma''(x)\models H''$, 
	\item\label{q3} $\sigma''(x)\models\phi(x)$ for $x\in var(u')$.
\end{inparaenum}
We note that using $\sigma''$ and rule $r''$, we then have, by (\ref{q1}) and (\ref{q2}), that 
$\sigma(t)\tran{a?}\sigma''(u')$ and $\sigma''(u')\models\varphi$ by (\ref{q3}) and the inductive hypothesis. 
So to complete the proof we are left to construct $\sigma''$. 

We start by setting $\sigma''(x)=\sigma(x)$ for each $x\in var(t)$. This ensures (\ref{q1}). Moreover, since 
$\sigma_r(x)=\sigma(x)=\sigma''(x)$ for each $x\in var(t)$ and $\sigma_r$ satisfies all the negative premises 
in $H$ by item (\ref{item:A}), by condition $\ref{item:neg}$ $\sigma''$ satisfies all the negative premises in $H''$. 
We now show how to extend the definition of $\sigma''$ in order to satisfy also the positive premises in $H''$.
\begin{itemize}
\item If $x_i\tran{b?}y\in H''$ and $y\notin var(u')$, then by condition~(\ref{item:pos}) above there is a positive premise 
$x_i\tran{b?}y'\in H$. Since $\sigma_r(x_i)=\sigma(x_i)=\sigma''(x_i)$ and $\sigma_r\models H$, there is some 
$p'$ such that $\sigma_r(x_i)\tran{b?}p'$. We set $\sigma''(y)=p'$.

\item If $x_i\tran{b?}y\in H''$ and $y\in var(u')$, then $\sigma''(x_i)=\sigma_r(x_i)=\sigma(x_i)\models\psi'(x)$. 
Now, by construction, $\psi'(x)$ has a conjunct $\eexis{b?}\phi(y)$. Again, by condition~(\ref{item:pos}, there is some 
$p'$ such that $\sigma_r(x_i)\tran{b?}p'\models\phi(y)$. We can therefore set $\sigma''(y)=p'$ to guarantee item~(\ref{q3}) above for $y$.
\end{itemize}
Finally, note that if $x\in var(t)\cap var(u')$, $\psi'(x)$ has $\phi(x)$ as conjunct and therefore setting 
$\sigma''(x)=\sigma(x)$ guarantees (\ref{q3}) for those variables. 

This completes the definition of $\sigma''$ having properties (\ref{q1})--(\ref{q3}), and we are done.
\end{proof}

\section{Proof of Theorem~\ref{thm:quiescent}}\label{App:thmquiescent}

We start by establishing the following lemma, which states the correctness of the definition of contradictory 
sets of formulae in Definition~\ref{def:contradicts}.

\begin{lemma}\label{Lem:contracorr}
Assume that $H_1$ and $H_2$ are {contradictory} sets of transition formulae whose left-hand sides are over 
distinct variables $x_1,\ldots,x_n$. Let $\sigma$ be a substitution mapping variables to states in an LTS that 
satisfies $H_1$. Then there is no substitution $\rho$ such that $\rho(x_i) = \sigma(x_i)$ for $1\leq i\leq n$ and 
$\rho$ satisfies $H_2$.
\end{lemma}

\begin{proof}
Since $H_1$ and $H_2$ are {contradictory}, there are $H_1' \subseteq H_1$ and $H_2' \subseteq H_2$ such 
that $H_1'$ and $H_2'$ contradict each other. We proceed by a case analysis on the possible form $H_1'$ 
and $H_2'$ may take, following Definition~\ref{def:contradicts}.
\begin{itemize}
\item Assume that $H_1'=\{x\tran{a}y\}$ and $H_2'=\{x\notran{a}\}$ for some $a\in L$. Since 
$\sigma(x)\tran{a}\sigma(y)$, there is no substitution $\rho$ such that $\rho(x) = \sigma(x) \notran{a}$.

\item Assume that $H_1'=\{x\tran{b!}y\}$ and $H_2'=\{x\tran{\delta!}z\}$ for some $b!\in O\setminus 
\{\delta!\}$.  Since $\sigma(x)\tran{b!}\sigma(y)$ and no state that can perform such a transition is quiescent, 
there is no substitution $\rho$ such that $\rho(x) = \sigma(x) \tran{\delta!}$. 

\item Assume that $H_1' \cup H_2' = \{x\notran{b!}\mid b!\in O\}$ for some variable $x$. Since $\sigma(x)$ satisfies $H_1'$ and every state in an LTS performs at least one output-labelled transition, we have that $\sigma(x)$ cannot satisfy $H_2'$ and we are done. 
\end{itemize}

The cases resulting by swapping the role of $H_1'$ and $H_2'$ in each item above are handled by symmetric arguments. 
\end{proof}

We are now ready to prove prove Theorem~\ref{thm:quiescent}. We will show the two implications separately.

\begin{proof}[\textbf{Proof of Theorem~\ref{thm:quiescent}}]
We assume that $f$ is quiescent consistent and prove Property~(\ref{eq:CBQ}):
\[
f(\arr{p})\tran{\delta!}p'\text{ iff } p'=f(\arr{p})\text{ and, for each  } a!\in O\setminus\{\delta!\},~ 
f(\arr{p})\notran{a!} . 
\]
	
\emph{Only if implication}: Assume that $f(\arr{p})\tran{\delta!}p'$, where $\arr{p}=p_1,\ldots,p_n$ is a 
sequence of states in some LTS. Then there are a rule
\begin{mathpar}
	\inferrule*[left={r}]
	{H}
	{f(\arr{x})\tran{\delta!}t}
\end{mathpar}
and a substitution $\sigma$ such that $\sigma(f(\arr{x}))=f(\arr{p})$, $\sigma(t)=p'$ and $\sigma$ satisfies 
the premises in $H$.  By condition~[$\delta_1$] in Definition~\ref{def:quiescent_consistent}, we have that 
$t=f(\arr{y})$ for some vector of variables $\arr{y}$ such that each $y_i$ is either $x_i$ or $x_i \tran{\delta!} 
y_i\in H$. By the requirement on $\tran{\delta!}$ in Definition~\ref{dfn:lts} and the fact that $\sigma$ 
satisfies the premises in $H$, we infer that $\sigma(y_i) = p_i$ for each $x_i \tran{\delta!} y_i\in H$. Thus $p' 
= \sigma(t)= \sigma(f(\arr{y})) = f(\arr{p})$. We are therefore left to show that $f(\arr{p})\notran{b!}$ for each 
$b!\in O\setminus\{\delta!\}$. To this end, it is enough to prove that if
\begin{mathpar}
	\inferrule*[left={r'}]
	{H'}
	{f(\arr{x})\tran{b!}t'}
\end{mathpar}
is a rule for $f$, then no substitution $\rho$ that agrees with $\sigma$ over $\arr{x}$ can satisfy $H'$. In 
order to see this, observe that, by condition~[$\delta_1$] in Definition~\ref{def:quiescent_consistent}, the 
sets $H$ and $H'$ are contradictory.  Since $\sigma$ satisfies $H$, by Lemma~\ref{Lem:contracorr}, no 
substitution $\rho$ that agrees with $\sigma$ over $\arr{x}$ can satisfy $H'$.\\  

\emph{If implication}: Assume now that $f(\arr{p})\notran{b!}$ for all $b!\in O\setminus\{\delta!\}$. We will 
argue that $f(\arr{p})\tran{\delta!}f(\arr{p})$. Let $\{r_1,\ldots r_n\}$ be the set of output-emitting rules for $f$ 
not having $\delta!$ as label of their conclusions. Because of our assumption above, for each $r_i$ there is 
some premise $l_i$ of $r_i$ for some variable $x_j$ in $\arr{x}$ that $p_j$ does not satisfy. For each $1\leq i 
\leq n$, we define $\overline{l}_i$ as follows: 
\[
\overline{l}_i = 
\begin{cases}
x_j \tran{a} y_i &\text{ if } l_i = x_j   \notran{a} \\ 
x_j   \notran{a}  &\text{ if } l_i = x_j \tran{a} y, \text{ for some $y$}. 
\end{cases}
\]

Here the variables $y_i$'s are all different and distinct from the variables in $\arr{x}$.  By 
condition~[$\delta_2$] in Definition~\ref{def:quiescent_consistent}, we have that the set of rules for $f$ 
includes the rule
\begin{mathpar}
	\inferrule*[left={r},right={\normalsize .}]
	{\{\overline{l}_1,\ldots \overline{l}_n\}}
	{f(\arr{x})\tran{\delta!}f(\arr{x})}
\end{mathpar}

Indeed, no two sets of formulae included in $\{\overline{l}_1,\ldots \overline{l}_n\}$ contradict each other 
because of the way that set of formulae is constructed.

It is now easy to see that there exists a substitution $\rho$ such that $\rho(f(\arr{x}))=f(\arr{p})$ and $\rho$ 
satisfies $\{\overline{l}_1,\ldots \overline{l}_n\}$. Hence, we conclude that $f(\arr{p})\tran{\delta!}f(\arr{p})$, 
and we are done.
\end{proof}

\end{document}